\documentclass[10pt,journal,compsoc]{IEEEtran}
\ifCLASSOPTIONcompsoc
  \usepackage[nocompress]{cite}
\else
  \usepackage{cite}
\fi
\ifCLASSINFOpdf
\else
\fi
\fontencoding{T1}
\usepackage[T1]{fontenc}
\usepackage{cite}
\usepackage{amsmath,amssymb,amsfonts}
\usepackage{amsthm}
\usepackage{graphicx}
\usepackage{textcomp}
\usepackage{xcolor}
\usepackage{lipsum}
\usepackage{subfig}
    
\usepackage[hyphens]{url}

\usepackage{hyperref}
\usepackage[capitalize]{cleveref}
\Crefname{figure}{Fig.}{Figs.}

\usepackage{algorithm}
\usepackage{algpseudocode}
\usepackage{enumitem}

\newcommand{\getsr}{{\:\stackrel{{\scriptscriptstyle \hspace{0.2em}\$}} {\leftarrow}\:}}
\newcommand{\concat}{\:\|\:}

\newtheorem{thm}{Theorem} %[section]

\usepackage{booktabs}
\usepackage{array}
\usepackage{wrapfig}

\begin{document}

\title{Denial-of-Service Vulnerability of Hash-based Transaction Sharding: Attack and Countermeasure}

\author{Truc~Nguyen
        and~My~T.~Thai% <-this % stops a space
\IEEEcompsocitemizethanks{\IEEEcompsocthanksitem T. Nguyen and My T. Thai are with the Department of Computer \& Information Science \& Engineering, University of Florida, Gainesville,
FL, 32611.\protect\\
% note need leading \protect in front of \\ to get a newline within \thanks as
% \\ is fragile and will error, could use \hfil\break instead.
E-mail: truc.nguyen@ufl.edu and mythai@cise.ufl.edu
% \IEEEcompsocthanksitem J. Doe and J. Doe are with Anonymous University.
}% <-this % stops an unwanted space
% \thanks{Manuscript received April 19, 2005; revised August 26, 2015.}
}

% The paper headers
% \markboth{Journal of \LaTeX\ Class Files,~Vol.~14, No.~8, August~2015}%
% {Shell \MakeLowercase{\textit{et al.}}: Bare Demo of IEEEtran.cls for Computer Society Journals}

% make the title area

\IEEEtitleabstractindextext{%
\begin{abstract}
Since 2016, sharding has become an auspicious solution to tackle the scalability issue in legacy blockchain systems. Despite its potential to strongly boost the blockchain throughput, sharding comes with its own security issues. To ease the process of deciding which shard to place transactions in, existing sharding protocols use a hash-based transaction sharding in which the hash value of a transaction determines its output shard. Unfortunately, we show that this mechanism opens up a loophole that could be exploited to conduct a single-shard flooding attack, a type of Denial-of-Service (DoS) attack, to overwhelm a single shard that ends up reducing the performance of the system as a whole. %Several experiments and analysis are conducted to demonstrate the impact as well as the practicality of our attack.

To counter the single-shard flooding attack, we propose a countermeasure that essentially eliminates the loophole by rejecting the use of hash-based transaction sharding. The countermeasure leverages the Trusted Execution Environment (TEE) to let blockchain's validators securely execute a transaction sharding algorithm with a negligible overhead. We provide a formal specification for the countermeasure and analyze its security properties in the Universal Composability (UC) framework. Finally, a proof-of-concept is developed to demonstrate the feasibility and practicality of our solution.
\end{abstract}

% Note that keywords are not normally used for peerreview papers.
}

\maketitle

% To allow for easy dual compilation without having to reenter the
% abstract/keywords data, the \IEEEtitleabstractindextext text will
% not be used in maketitle, but will appear (i.e., to be "transported")
% here as \IEEEdisplaynontitleabstractindextext when the compsoc 
% or transmag modes are not selected <OR> if conference mode is selected 
% - because all conference papers position the abstract like regular
% papers do.
\IEEEdisplaynontitleabstractindextext
% \IEEEdisplaynontitleabstractindextext has no effect when using
% compsoc or transmag under a non-conference mode.

% For peer review papers, you can put extra information on the cover
% page as needed:
% \ifCLASSOPTIONpeerreview
% \begin{center} \bfseries EDICS Category: 3-BBND \end{center}
% \fi
%
% For peerreview papers, this IEEEtran command inserts a page break and
% creates the second title. It will be ignored for other modes.
\IEEEpeerreviewmaketitle

\IEEEraisesectionheading{\section{Introduction}\label{sec:introduction}}
% \vspace*{-2mm}
% Computer Society journal (but not conference!) papers do something unusual
% with the very first section heading (almost always called "Introduction").
% They place it ABOVE the main text! IEEEtran.cls does not automatically do
% this for you, but you can achieve this effect with the provided
% \IEEEraisesectionheading{} command. Note the need to keep any \label that
% is to refer to the section immediately after \section in the above as
% \IEEEraisesectionheading puts \section within a raised box.

% The very first letter is a 2 line initial drop letter followed
% by the rest of the first word in caps (small caps for compsoc).
% 
% form to use if the first word consists of a single letter:
% \IEEEPARstart{A}{demo} file is ....
% 
% form to use if you need the single drop letter followed by
% normal text (unknown if ever used by the IEEE):
% \IEEEPARstart{A}{}demo file is ....
% 
% Some journals put the first two words in caps:
% \IEEEPARstart{T}{his demo} file is ....
% 
% Here we have the typical use of a "T" for an initial drop letter
% and "HIS" in caps to complete the first word.
Sharding, an auspicious solution to tackle the scalability issue of blockchain, has become one of the most trending research topics and been intensively studied in recent years \cite{luu2016secure,zamani2018rapidchain,kokoris2018omniledger,wang2019monoxide,nguyen2019optchain,huang2020repchain,yu2020ohie,hong2021pyramid,manuskin2020ostraka}. In the context of blockchain, sharding is the approach of partitioning the set of nodes (or validators) into multiple smaller groups of nodes, called shard, that operate in parallel on disjoint sets of transactions and maintain disjoint ledgers. By parallelizing the consensus work and storage, sharding reduces drastically the storage, computation, and communication costs that are placed on a single node, thereby scaling the system throughput proportionally to the number of shards. Previous studies \cite{zamani2018rapidchain,kokoris2018omniledger,nguyen2019optchain} show that sharding could potentially improve blockchain's throughput to thousands of transactions per second (whereas the current Bitcoin system only handles up to 7 transactions per second and requires 60 minutes confirmation time for each transaction).

Despite the incredible results in improving the scalability, blockchain sharding is still vulnerable to some severe security problems. The root of those problems is that, with partitioning, the honest majority of mining power or stake share is dispersed into individual shards. This significantly reduces the size of the honest majority in each shard, which in turn dramatically lowers the attack bar on a specific shard. Hence, a blockchain sharding system must have some mechanisms to prevent adversaries from gaining the majority of validators of a single shard, this is commonly referred to as single-shard takeover attack.% \cite{ethereum}.

In this paper, we take a novel approach by exploiting the inter-shard consensus to identify a new vulnerability of blockchain sharding. One intrinsic attribute of blockchain sharding is the existence of cross-shard transactions that, simply speaking, are transactions that involve multiple shards. These transactions require the involved shards to perform an inter-shard consensus mechanism to confirm the validity. Hence, intuitively, if we could perform a Denial-of-Server (DoS) attack to one shard, it would also affect the performance of other shards via the cross-shard transactions. Furthermore, both theoretical and empirical analysis \cite{zamani2018rapidchain,nguyen2019optchain} show that most existing sharding protocols have 99\% cross-shard transactions. This implies that an attack on one shard could potentially impact the performance of the entire blockchain. In addition, with this type of attacks, the attacker is a client of the blockchain system, hence, this attack can be conducted even when we can guarantee the honest majority in every shard.

Although existing work does have some variants of flooding attacks that try to overwhelm the entire blockchain by having the attacker generate a superfluous amount of dust transactions \cite{saad2019mempool,baqer2016stressing}, it is unclear how we could conduct this attack in a sharding system. In fact, we emphasize that a conventional transactions flooding attack on the entire blockchain (as opposed to a single shard) would not be effective for two reasons. First, blockchain sharding has high throughput, hence, the cost of attack would be enormous to generate a huge amount of dust transactions that is sufficiently much greater than the system throughput. Second, more importantly, since the sharding system scales with the number of shards, it can easily tolerate such attacks by adding more shards to increase throughput.

To bridge this gap, we propose a single-shard flooding attack to exploit the DoS vulnerability of blockchain sharding. Instead of overwhelming the entire blockchain, an attacker would strategically place a tremendous amount of transactions into one single shard in order to reduce the performance of that shard, as the throughput of one shard is not scalable. The essence of our attack comes from the fact that most sharding proposals use hash-based transaction sharding \cite{kokoris2018omniledger,zamani2018rapidchain,luu2016secure,huang2020repchain}: a transaction's hash value is used to determine which shard to place the transaction (i.e., output shard). Since that hash value (e.g., SHA-256) of a transaction is indistinguishable from that of a random function, this mechanism can efficiently distribute the transactions evenly among the shards and thus widely adopted. %Furthermore, because the hash value is an attribute of the transaction, the validators do not have to reach consensus on where to place a transaction.
Therefore, an attacker can manipulate the hash to generate an excessive amount of transactions to one shard.

As we argue that using hash values to determine the output shard is not secure, we propose a countermeasure to efficiently eliminate the attack for the sharding system. By not using the transaction's hash value or any other attributes of the transaction, we can delegate the task of determining the output shard to the validators, then the adversary cannot carry out this DoS attack. However, this raises two main challenges: (1) what basis can be used to determine the output shard of a transaction, and (2) how honest validators can agree on the output of (1). For the first challenge, we need a \textit{transaction sharding algorithm} to decide the output shard for each transaction. OptChain \cite{nguyen2019optchain} is an example algorithm where it aims to minimize the number of cross-shard transactions and also balance the load among the shards. For the second challenge, a naive solution is to have the validators reach on-chain consensus on the output shard of every transaction. However, that would be very costly and reject the main concept of sharding, that is, each validator only processes a subset of transactions to parallelize the consensus work and storage.

To overcome the aforementioned challenge, we establish a system for executing the transaction sharding algorithm off-chain and attesting the correctness of the execution. As blockchain validators are untrusted, we need to guarantee that the execution of the transaction sharding algorithm is tamper-proof. To accomplish this, we leverage the Trusted Execution Environment (TEE) to isolate the execution of the algorithm inside a TEE module, shielding it from potentially malicious hosts. With this approach, we are not imposing any significant on-chain computation overhead as compared to the hash-based transaction sharding and also maintain the security properties of the blockchain. Moreover, this solution can be easily integrated into existing blockchain sharding proposals, and as modern Intel CPUs from 2014 support TEE, the proposed countermeasure is compatible with current blockchain systems.

{\bf Contribution.} Our main contributions are as follows:
\begin{itemize}
    \item We identify a new attack on blockchain sharding that exploits the loophole of using hash-based transaction sharing, namely single-shard flooding attack. %, including its threat model and analysis of its damage on the system performance. %Specifically, we present the threat model and specify how to conduct the attack in detail. Then, we perform some analysis of the attack to demonstrate its impact on the system performanceim.
    \item To evaluate the potential impact of this attack on the blockchain system, we develop a discrete-event simulator for blockchain sharding that can be used to observe how sharding performance changes when the system is under attack. Not only for our attack analysis purposes, this simulator can also assist the research community in evaluating the performance of a sharding system without having to set up multiple computing nodes. 
    % The simulator especially enables easy implementation of different transaction sharding algorithms as well as cross-shard validation mechanisms.
    \item We propose a countermeasure to the single-shard flooding attack by executing transaction sharding algorithms using TEE. Specifically, we provide a formal specification of the system and formally analyze its security properties in the Universal Composability (UC) framework with a strong adversarial model.
    \item To validate our proposed countermeasure, we develop a proof-of-concept implementation of the system and provide a performance analysis to demonstrate its feasibility.
\end{itemize}

\noindent {\bf Organization.} The rest of the paper is structured as follows. Some background and related work are summarized in \Cref{sec:relate}. \Cref{sec:Attack} describes in detail the single-shard flooding attack with some preliminary analysis to demonstrate its practicality. In \Cref{sec:analyze}, we present the construction of our simulator and conduct some experiments to demonstrate the damage of the attack. The countermeasure is discussed in \Cref{sec:countermeasure} with a formal specification of the system. \Cref{sec:sec-eval} gives a security analysis of the countermeasure along with a performance evaluation on the proof-of-concept implementation. Finally, \Cref{sec:con} concludes our paper.

\section{Background and Related Work} \label{sec:relate}
% In this section, we establish some background on the blockchain sharding as well as examine some prior work that is related to ours.

{\bf Blockchain sharding.}
Several solutions \cite{hong2021pyramid,manuskin2020ostraka,yu2020ohie,zamani2018rapidchain,nguyen2019optchain,luu2016secure,kokoris2018omniledger,wang2019monoxide,huang2020repchain} suggest partitioning the blockchain into shards to address the scalability issue in Bitcoin blockchain. Typically, with sharding, the blockchain's state is divided into multiple shards, each has its own independent state and transactions and is managed by the shard's validators. By having multiple shards where each of them processes a disjoint set of transactions, the computation power is parallelized, and sharding in turn helps boost the system throughput with respect to the number of shards.  %With the exception of Ethereum sharding \cite{ethereum}, most of existing sharding protocols are developed on top of Bitcoin blockchain. 
Some main challenges of a sharding protocol include (1) how to securely assign validators to shards, (2) intra-shard consensus, (3) assigning transactions to shards, and (4) processing cross-shard transactions. Our proposed attack exploits the third and fourth challenges of sharding that deal with transactions in a sharding system.

In a simple manner, a transaction is cross-shard if it requires confirmations from more than one shard. In the Unspent Transaction Outputs (UTXO) model used by Bitcoin, each transaction has multiple outputs and inputs where an output dictates the amount of money that is sent to a Bitcoin address. Each of the outputs can be used as an input to another transaction. To prevent double-spending, an output can be used only once. Denote $tx$ as a transaction with two inputs $tx_1$ and $tx_2$, this means $tx$ uses one or more outputs from transaction $tx_1$ and $tx_2$. Let $S_{1}, S_2,$ and $S_3$ be the shards containing $tx_1, tx_2,$ and $tx$, respectively, we refer $S_1$ and $S_2$ as the \textit{input shards} of $tx$, and $S_3$ as the \textit{output shard}. If these three shards are the same, $tx$ is an in-shard transaction, otherwise $tx$ is cross-shard.

To determine the output shard of a transaction, most sharding protocols use the hash value of the transaction to calculate the ID of the output shard. By leveraging the hash value, the transactions are effectively assigned to shards in a uniformly random manner. However, we show that this mechanism can be manipulated to perform a DoS attack.

To process cross-shard transactions, several cross-shard validation mechanisms have been proposed \cite{kokoris2018omniledger,zamani2018rapidchain,wang2019monoxide}. A cross-shard validation mechanism determines how input and output shards can coordinate to validate a cross-shard transaction. This makes the process of validating cross-shard transactions particularly expensive since the transaction must wait for confirmations from all of its input shards before it can be validated in the output shard. \Cref{fig:crosstx} illustrates this process. Our attack takes advantage of this mechanism to cause a cascading effect that creates a negative impact on shards that are not being attacked.

\begin{figure}
    \centering
    \includegraphics[width=\linewidth]{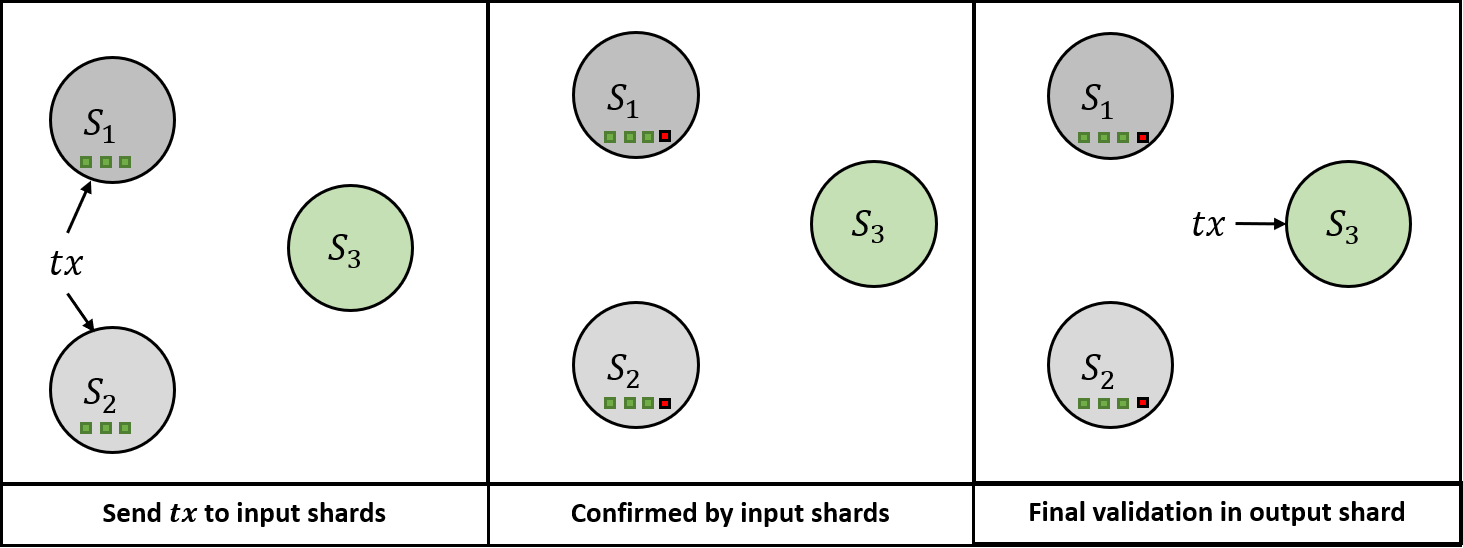}
    \caption{Processing cross-shard transaction $tx$. $tx$ has input shards $S_1$, $S_2$, and output shard $S_3$. $tx$ has to wait to be confirmed in $S_1, S_2$ %the input shards 
    before it can be validated in $S_3$.%the output shard
    }
    \label{fig:crosstx}
\end{figure}

{\bf Denial-of-service and flooding attacks.}
Denial-of-service (DoS) is commonly defined as an intentional attack on availability and it has been around for decades \cite{mirkovic2004taxonomy}. In general, it is hard to defend against the DoS attacks and we typically can only mitigate them. A common mitigation technique is based on anomaly detection to filter out DoS attack packets \cite{tan2013system,YU20084212,biron2018real}.

In the context of Bitcoin, DoS typically takes the form of a flooding attack that overwhelms the system with a flood of transactions. Over the years, we have observed the economic impact of this attack as Bitcoin has been flooded with dust transactions by malicious users to make legitimate users pay higher mining fees \cite{btc_2018}. 
% In November of 2017, the Bitcoin mempool size exceeded 115,000 unconfirmed transactions, the value of these unconfirmed transactions sum up to 110,611 BTC, worth over 900 million US dollars, weighting about 36 MB \cite{ccn.com_2017}. In June 2018, the Bitcoin's mempool was flooded with 4500 unconfirmed dust transactions which eventually increased the mempool size to 45MB. As a result, the mining fee was greatly increased and it caused legitimate users to pay higher fee to get their transactions confirmed \cite{btc_2018}. Therefore, it is extremely critical to study this attack rigorously.
There exist some variants of flooding attack that aim to overwhelm an entire blockchain system \cite{saad2019mempool,baqer2016stressing}, not a single shard. The main concept of the attack is to send a huge amount of transactions to overwhelm the mempool, fill blocks to their maximum size, and effectively delay other transactions. Typically, unconfirmed transactions are stored in the mempools managed by blockchain validators. In contrast to the limited block size, the mempool size has no size limit.

% Despite the severely negative impact, research conducted on this attack is limited. In \cite{saad2019mempool}, the authors study a low cost flooding attack on Blockchain applications in which the adversary generates Sybil nodes that can flood the mempools with unconfirmed transactions. By paying the minimum relaying fee without paying the minimum mining fee, the transactions would be stored in the mempool but might not get mined, thus the attacker would not lose the transaction fee to perform the attack. In another work by Bager et al. \cite{baqer2016stressing}, the vulnerability of Bitcoin to flooding attacks is exploited. Malicious users leverage the 1MB block size limit in Bitcoin to attack the blockchain with low-value dust transactions, which causes enormous delay in validating legitimate transactions.

This kind of attack requires the attacker to flood the blockchain system at a rate that is much greater than the system throughput. Intuitively, such an attack is not effective on a sharding system because its throughput is exceedingly high. 
% In fact, a simple solution to the flooding attack is to add more shards to increase the overall throughput since the system scales with the number of shards. In contrast, the throughput of one shard is not scalable, thus, it would be more reasonable to attack a single shard and make it become the performance bottleneck of the whole system. 
In this paper, we show how attackers can manipulate the transaction’s hash to overwhelm a single shard, thereby damaging the entire blockchain through cascading effects caused by cross-shard transactions. 

{\bf Blockchain on Trusted Execution Environment (TEE).}
A key building block of our countermeasure is TEE. Memory regions in TEE are transparently encrypted and integrity-protected with keys that are only available to the processor. TEE's memory is also isolated by the CPU hardware from the rest of the host's system, including high-privilege system software. Thus the operating system, hypervisor, and other users cannot access the TEE’s memory. 
% There have been multiple available implementations of TEE including Intel SGX \cite{intel2014software}, ARM TrustZone \cite{armltd}, and Keystone \cite{team}. 
Among available implementations of TEE, Intel SGX \cite{intel2014software} supports generating remote attestations that are used to prove the correct execution of programs running inside TEE.

There has been a recent growth in adopting TEEs to improve blockchains \cite{matetic2019bite,lind2019teechain,cheng2019ekiden,das2019fastkitten}, but not sharding systems. Teechain \cite{lind2019teechain} proposes an improvement over the off-chain payment network in Bitcoin using TEE to enable asynchronous blockchain access. BITE \cite{matetic2019bite} leverages TEE to further enhance the privacy of Bitcoin's clients. In \cite{das2019fastkitten,cheng2019ekiden}, the authors develop secure and efficient smart contract platforms on Bitcoin and Ethereum, respectively, using TEE as a module to execute the contract's code.

We argue that TEE can be used to develop an efficient countermeasure for the single-shard flooding attack in which transaction sharding algorithms can be securely executed inside a TEE module. However, since existing solutions are designed to address some very specific issues such as smart contracts or payment networks, applying them to blockchain {\em sharding} systems is not straightforward.

\section{Single-shard flooding attack}\label{sec:Attack}
In this section, we describe our proposed single-shard flooding attack on blockchain sharding starting with the threat model and detail on performing the attack. Then, we present some preliminary analysis of the attack to illustrate its potential impact and practicality.

\subsection{Threat Model} \label{ssec:threat}
{\bf Attacker.} We use Bitcoin-based sharding systems, such as OmniLedger, RapidChain, and Elastico, as the attacker's target. We consider an attacker who is a client of the blockchain system such that: 
\begin{enumerate}
    \item The attacker possesses enough amount of Bitcoin addresses to perform the attack. In practice, Bitcoin addresses can be generated at no cost.
    \item The attacker has spendable bitcoins in its wallet and the balance is large enough to issue multiple transactions between its addresses for this attack. Each of the transactions is able to pay the minimum relay fee $minRelayTxFee$. We will discuss the detailed cost in the next section.
    \item The attacker is equipped with software that is capable of generating transactions at a rate that is higher than a shard's throughput, which will be discussed in the subsequent section.
    \item Since this is a type of DoS attack, to prevent it from being blocked by the blockchain network, the attacker can originate the attack from multiple sources. The attacker can also leverage some kind of anonymous communication when connecting to the Bitcoin network to prevent the network packets from revealing the attacker’s identity. Therefore, we assume that the attacker can remain anonymous and untraceable.
\end{enumerate}

{\bf Goals of attacks.}
By employing the concept of flooding attack, the main goal of the single-shard flooding attack is to overwhelm a single shard by sending a huge amount of transactions to that shard. The impact of this attack has been widely studied on non-sharded blockchains like Bitcoin such that it can reduce the system performance by delaying the verification of legitimate transactions and eventually increase the transaction fee.

Furthermore, with the concept of cross-shard transactions where each transaction requires confirmation from multiple shards, an attack to overwhelm one shard could affect the performance of other shards and reduce the system's performance as a whole. For example, in \Cref{fig:crosstx}, if $S_1$ were under attack, the transaction validation would also be delayed in $S_3$. Previous work \cite{zamani2018rapidchain} shows that placing transactions using their hash value could result in 99.98\% of cross-shard transactions. Since the throughput of one shard is limited, under our attack, it could effectively become the performance bottleneck of the whole system. Therefore, to make the most out of this scheme, the attacker would target the shard that has the lowest throughput in the system.

{\bf How to perform attack.}
In most Bitcoin-based sharding systems, such as OmniLedger, RapidChain, and Elastico, the hash of a transaction determines which shard to put the transaction in. Specifically, the ending bits of the hash value indicate the output shard ID. The main idea of our attack is to have the attackers manipulate the transaction's hash in order to place it into the shard that they want to overwhelm. To accomplish this, we conduct a brute-force generation of transactions by alternating the output addresses of a transaction until we find an appropriate hash value. 

Let $T$ be the shard that the attacker wants to overwhelm. We define a "malicious transaction" as a transaction whose hash was manipulated to be put in shard $T$. Denote $tx$ as a transaction, $tx.in$ is the set of input addresses, and $tx.out$ is the set of output addresses. We also denote $O$ as the set of attacker's addresses, $I \subseteq O$ as the set of attacker's addresses that are holding some bitcoins. Let $\mathcal{H}(\cdot)$ be the SHA-256 hash function (its output is indistinguishable from that of a random function), \Cref{algo:tx} describes how to generate a malicious transaction in a system of $2^N$ shards.

Starting with a raw transaction $tx$, the algorithm randomly samples a set of input addresses for $tx.in$ from $I$ such that the balance of those addresses is greater than the minimum relay fee. It then randomly samples a set of output addresses for $tx.out$ from $O$ and set the values for $tx.out$ so that $tx$ can pay the minimum relay fee. The hash value of the transaction is determined by double hashing the transaction's data using the SHA-256 function. It checks if the final $N$ bits indicate $T$ ($\mathbin{\&}$ denotes a bitwise AND), if that is true, it outputs the malicious $tx$ that will be placed into shard $T$. Otherwise, it re-samples another set of output addresses for $tx.out$ from $O$.

\begin{algorithm}
    \caption{Generate a malicious transaction}
    \label{algo:tx}
    \hspace*{\algorithmicindent} \textbf{Input:} $\mathcal{I}, \mathcal{O}, N, T$ \\
    \hspace*{\algorithmicindent} \textbf{Output:} A malicious transaction $tx$
	\begin{algorithmic}[1]
	    \State $tx \gets$ raw transaction
	    \State $tx.in \getsr \mathcal{I}$
		\While {$\mathcal{H}(\mathcal{H}(tx)) \mathbin{\&} (\mathbin{1}^{256-N} \:\|\: \mathbin{0}^{N}) \neq T$}
		    \State $tx.out \getsr \mathcal{O}$
		    \State Set values for $tx.out$ to satisfy the $minRelayTxFee$.
		\EndWhile
		\State \textbf{Ret} $tx$
	\end{algorithmic}
\end{algorithm}

\subsection{Preliminary Analysis} \label{ssec:analysis}
\subsubsection{Capability of generating malicious transactions}
In this section, we demonstrate the practicality of the attack by assessing the capability of generating malicious transactions on a real machine. Suppose we have $2^N$ shards and a transaction $tx$, that means we will use $N$ ending bits of $\mathcal{H}(tx)$ to determine its shard. Suppose we want to put all transactions into shard 0, we need to generate some malicious transactions $tx$ such that the last $N$ bits of $\mathcal{H}(tx)$ must be 0. We calculate the probability of generating a malicious transaction as follows. As a SHA-256 hash has 256 bits, the probability of generating a hash with $N$ ending zero bits will be $\frac{2^{256-N}}{2^{256}}=\frac{1}{2^N}$. Therefore, we expect to obtain 1 malicious transaction per generating $2^N$ transactions. That means if we have 16 shards, we can obtain 1 malicious tx (i.e., the last 4 bits are zero) per generating 16 transactions.

To see the capability of generating malicious transactions, we conduct an experiment on an 8th generation Intel Core i7 laptop. The program to generate transactions is written in C++ and runs with 8 threads. When the number of shards is 64, the program can generate up to 1,644,736 transaction hashes per second, of which there are 26,939 malicious transactions (8 ending bits are zero). In short, within 1 second, a laptop can generate about 26,939 malicious transactions, which is potentially much more than the throughput of one shard. \Cref{tab:maltx} shows the number of malicious transactions generated per second with respect to the number of shards. Note that, in practice, an attacker can easily produce much higher numbers by using a more highly capable machine with a faster CPU.

\begin{table}[]
\centering
\caption{Capability of generating malicious tx with respect to the no. of shards on an Intel Core i7 laptop with 8 threads.}
\label{tab:maltx}
\begin{tabular}{@{}cc@{}}
\toprule
\multicolumn{1}{l}{No. of shards} & \multicolumn{1}{l}{No. of malicious tx per sec} \\ \midrule
2                                 & 823,512                                         \\
4                                 & 412,543                                         \\
8                                 & 205,978                                         \\
16                                & 103,246                                         \\
32                                & 52,361                                          \\
64                                & 26,939                                          \\ \bottomrule
\end{tabular}
\end{table}

\subsubsection{Cost of attacks}  The default value of $minRelayTxFee$ in Bitcoin is 1,000 satoshi per kB, which is about \$0.10 (as of Feb 2020). Taking into account that the average transaction size is 500 bytes, each transaction needs to pay \$0.05 as the $minRelayTxFee$. Our experiments below show that generating 2,500 malicious transactions is enough to limit the throughput of the whole system by that of the attacked shard. Hence, the attacker needs about \$125 to perform the attack effectively. Furthermore, by paying the minimum relay fee without paying the minimum transaction fee, the malicious transactions will still be relayed to the attacked shard's mempool but will not be confirmed, thereby retaining the starting balance.

\subsubsection{Cascading effect of the single-shard flooding attack}
We estimate the portion of transactions that are affected by the attack. A transaction is affected if one of its input shards or the output shard is the attacked shard. In \cite{manuskin2020ostraka}, the authors calculate the ratio of transactions that could be affected when one shard is under attack. Specifically, when considering a system with $n$ shards and transactions with $m$ inputs, the probability of a transaction to be affected by the attack is $1-(\frac{n-1}{n})^{m+1}$. The result is illustrated in \Cref{fig:impact}. As can be seen, a typical transaction with 2 or 3 inputs has up to 70\% chance of being affected with 4 shards. However, with 16 shards about 20\% of the transactions are still affected. Note that this number only represents the transactions that are "directly" affected by the attack, the actual number is higher when considering transactions that depend on the delayed transactions. 

\begin{figure}
	\centering
	\includegraphics[width=0.7\linewidth]{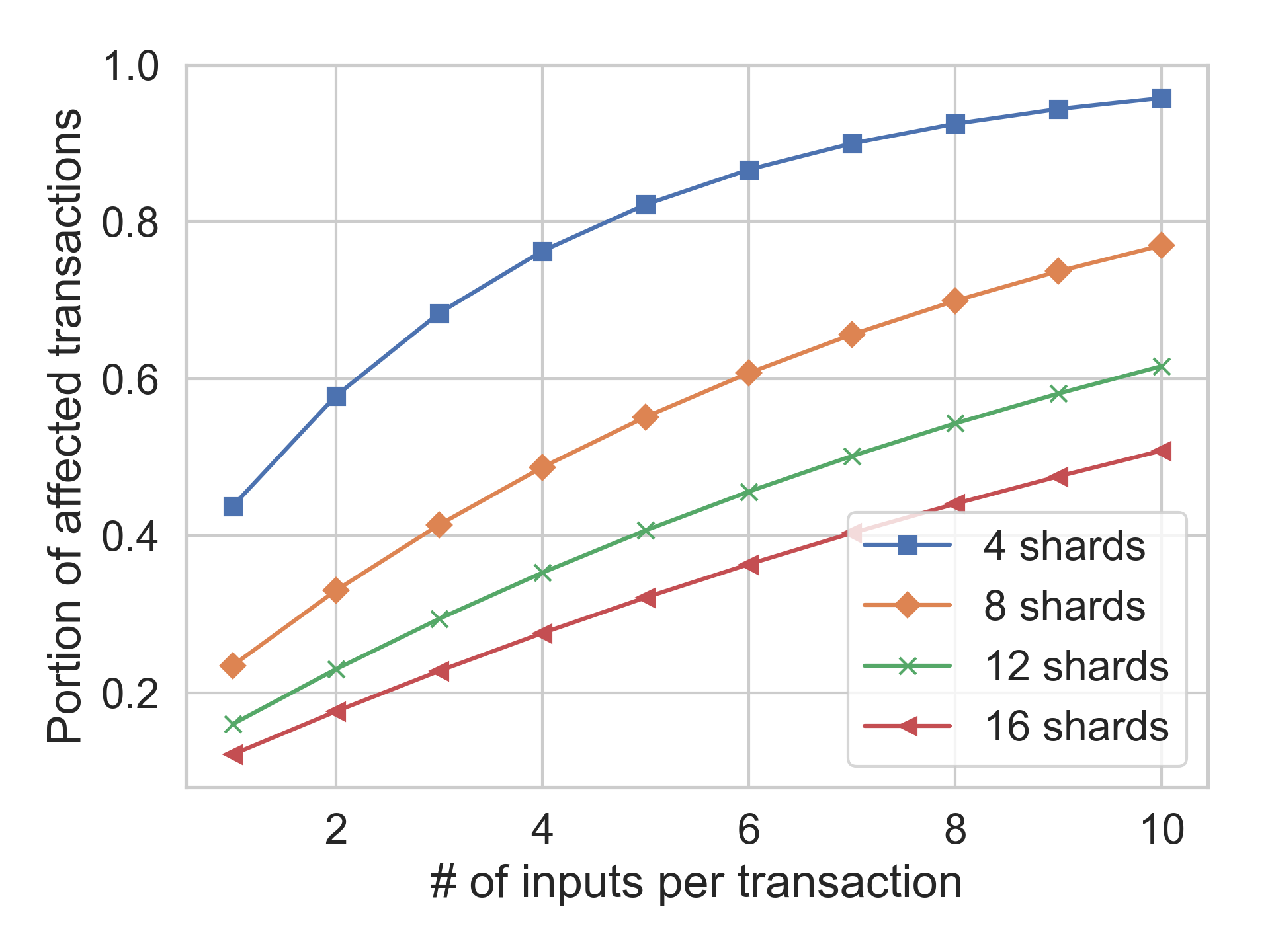}
	\caption{Affected transactions by the single-shard flooding attack}
	\label{fig:impact}
\end{figure}

Even though the analysis shows that the number of affected transactions is less at 16 shards than at 4 shards, in fact, the attack does much more damage to the 16-shard system. The intuition of this scenario is that as we increase the number of shards, we also increase the number of input shards per transaction. Since a transaction has to wait for the confirmations from all of its input shards, an affected transaction in the 16-shard system takes more time to be validated than that in the 4-shard system. The experiments in the next section will illustrate this impact in more detail.

\section{Analyzing the Attack's Impacts} \label{sec:analyze}
In this section, we present a detailed analysis of our attack, especially how it impacts the system performance as a whole. Before that, we describe the design of our simulator that is developed to analyze the performance of a blockchain sharding system.

\subsection{Simulator}
Our implementation was based on SimBlock \cite{aoki2019simblock}, a discrete-event Bitcoin simulator that was designed to test the performance of the Bitcoin network. SimBlock is able to simulate the geographical distribution of Bitcoin nodes across six regions (North America, South America, Europe, Asia, Japan, Australia) of which the bandwidth and propagation delay are set to reproduce the actual behavior of the Bitcoin network. Nevertheless, SimBlock fails to capture the role of transactions in the simulation, which is an essential part in evaluating the performance of blockchain sharding systems.

Our work improves SimBlock by taking into consideration the Bitcoin transactions and simulating the behavior of sharding. \Cref{fig:class} shows a simple UML class diagram depicting the relations between components of our simulator. As can be seen later, our simulator can be easily used to evaluate the performance of any existing or future sharding protocols.

\begin{figure}
	\centering
	\includegraphics[width=0.7\linewidth]{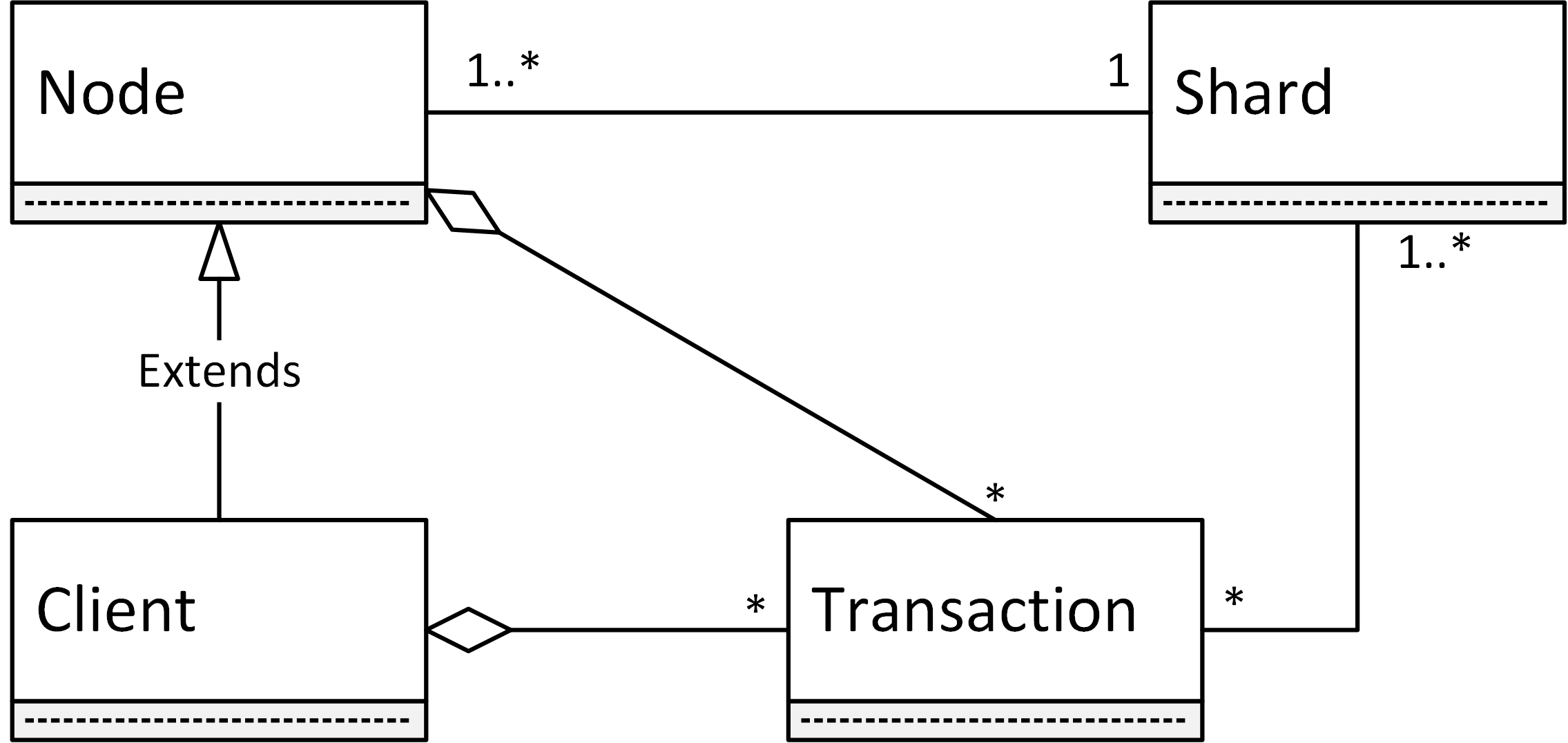}
	\caption{UML class diagram of the simulator}
	\label{fig:class}
\end{figure}

\textbf{Transactions.} The sole purpose of SimBlock was only to show the block propagation so the authors did not consider transactions. To represent Bitcoin transactions, we adopt the Transaction-as-Nodes (TaN) network proposed in \cite{nguyen2019optchain}. Each transaction is abstracted as a node in the TaN network, there is a directed edge $(u,v)$ if transaction $u$ uses transaction $v$ as an input. In our simulator, each transaction is an instance of a Transaction class and can be directly obtained from the Bitcoin dataset. At the beginning of the simulation, a Client instance loads each transaction from the dataset and sends them to the network to be confirmed by Nodes. Depending on the transaction sharding algorithm, each transaction could be associated with one or more shards.

Furthermore, our simulator can also emulate real Bitcoin transactions in case we need more transactions than what we have in the dataset or we want to test the system with a different set of transactions. With regard to sharding, the two important factors of a transaction are the degree and the input shards. From the Bitcoin dataset of more than 300 million real Bitcoin transactions, we fit the degree distribution with a power-law function as in \Cref{fig:degree} (black dots are the data, and the blue line is the resulting power-law function). The resulting function is $y = 10^{6.7}x^{-2.3}$. \Cref{fig:inputshards} shows the number of input shards with 16 shards (using hash-based transaction sharding) that is also fitted with a power-law function. The resulting function is $y = 10^{7.2}x^{-2.2}$. Hence, the Client can use these distributions to sample the degree and input shards when generating transactions that resemble the distribution of the real dataset.

% \begin{figure}
% 	\centering
% 	\includegraphics[width=0.7\linewidth]{Degree.png}
% 	\caption{Degree distribution of the Bitcoin transactions. The black dots are the data, the blue line shows a fitted power-law $y = 10^{6.7}x^{-2.3}$.}
% 	\label{fig:degree}
% \end{figure}

% \begin{figure}
% 	\centering
% 	\includegraphics[width=0.8\linewidth]{inshard.png}
% 	\caption{Distribution of Bitcoin transactions' number of input shards. The black dots are the data, the blue line shows a fitted power-law $y = 10^{7.2}x^{-2.2}$.}
% 	\label{fig:inputshards}
% \end{figure}

\begin{figure}
	\centering
% 	\hspace{-20px}
	\subfloat[Degree distribution of the Bitcoin transactions. Fitted power-law $y = 10^{6.7}x^{-2.3}$.]{
	    \includegraphics[width=0.45\linewidth]{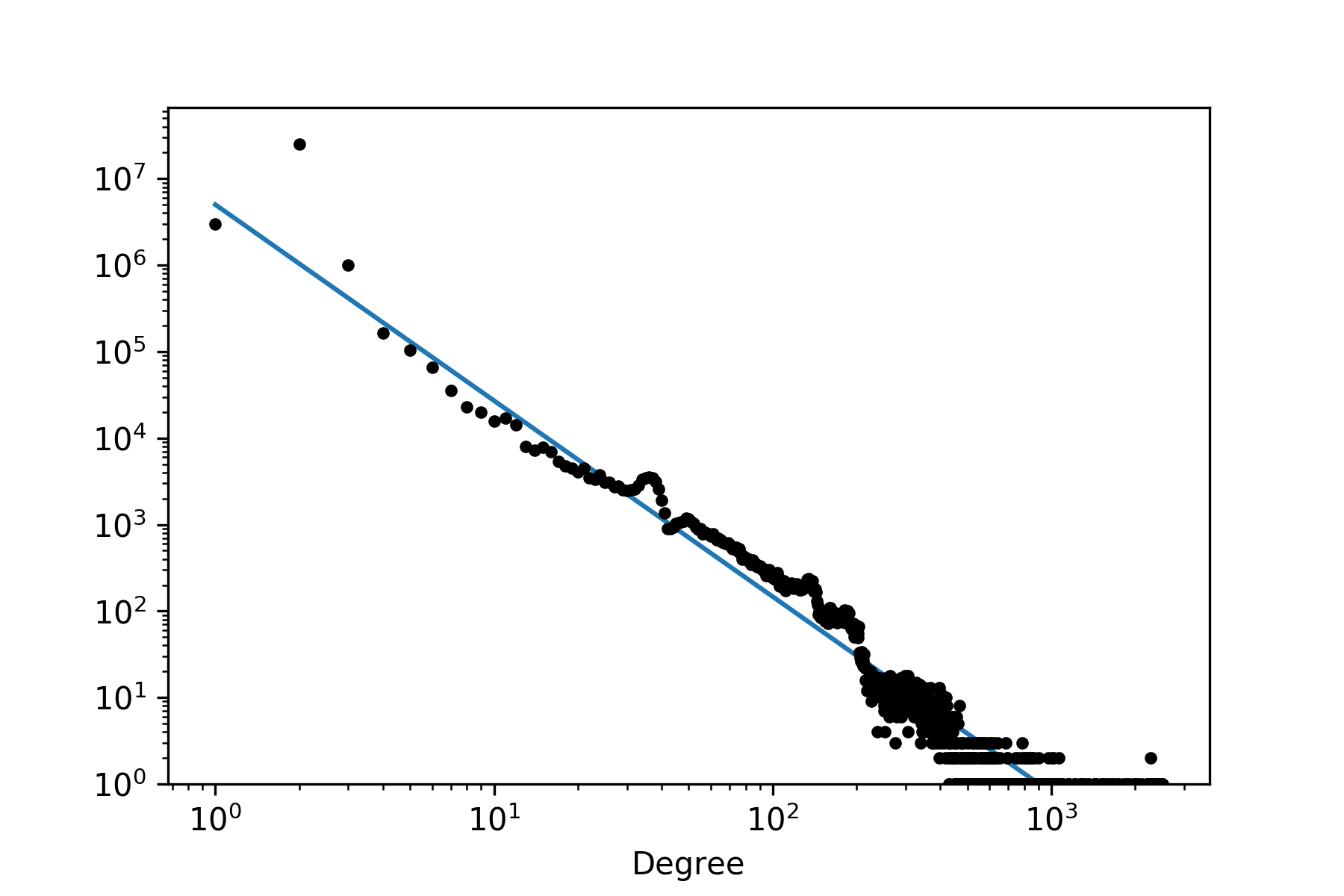}
    	\label{fig:degree}
	}%
	\hfill
	\subfloat[Distribution of Bitcoin transactions' number of input shards. Fitted power-law $y = 10^{7.2}x^{-2.2}$.]{
	    \includegraphics[width=0.45\linewidth]{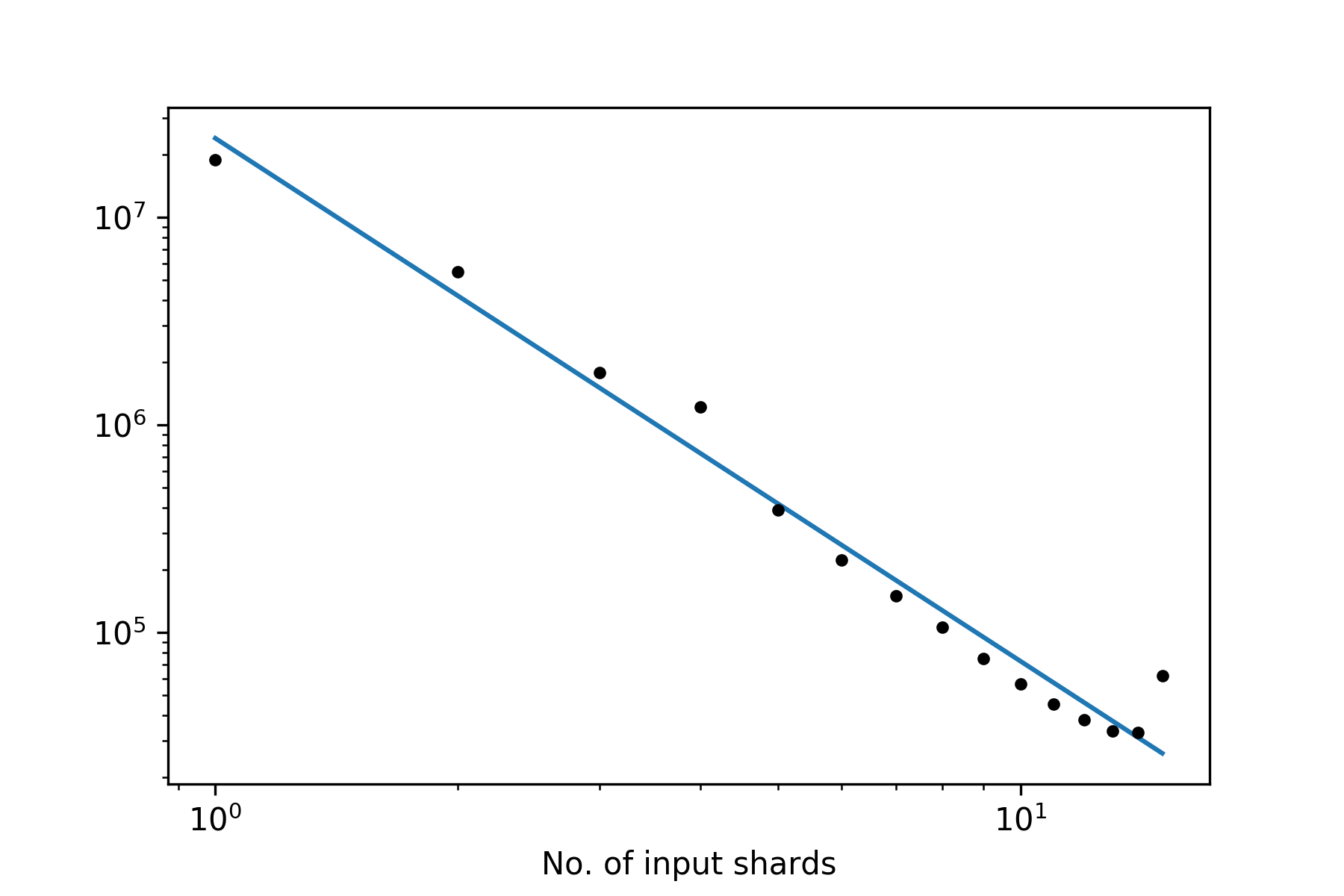}
	    \label{fig:inputshards}
    }
	\caption{Bitcoin transactions. The black dots are the data, the blue line shows a fitted power-law function}
	\label{fig:bitcontx}
% 	\vspace{-3.5mm}
\end{figure}

\textbf{Sharding.} After the simulator generates Node instances, each of them is distributed into an instance of Shard. All nodes in a shard share the same ledger and a mempool of unconfirmed/pending transactions. In the class Node, we implement a cross-shard validation algorithm that decides how nodes in different shards can communicate and confirm cross-shard transactions. In the current implementation, we use the mechanism proposed in \cite{kokoris2018omniledger} to process cross-shard transactions. 

For each Node instance, upon receiving a transaction, it will relay the transaction to the destination shard. When the transaction reaches the shard, it will be stored in the mempool of the Node instances in that shard. Each transaction in the mempool is then validated using an intra-shard consensus protocol.

\textbf{Use cases of the simulator.} Besides being used to test the impact of our proposed attack, researchers can also use the simulator to evaluate the performance of multiple blockchain sharding systems. By far, most experiments on blockchain sharding have to be run on numerous rented virtual machines \cite{zamani2018rapidchain,kokoris2018omniledger,luu2016secure,wang2019monoxide}, this is notably costly and complicated to set up. Without having to build the whole blockchain system, our simulator is particularly useful when researchers need to test various algorithms and system configurations on blockchain long before deploying the real system. 

By using simulation, various setups can be easily evaluated and compared, thereby making it possible to recognize and resolve problems without the need of performing potentially expensive field tests. By exploiting a pluggable design, the simulator can be easily reconfigured to work with different algorithms on transaction sharding, cross-shard validation, validators assignment, and intra-shard consensus protocol.

\subsection{Experimental Evaluations}
Our experiments are conducted on 10 million real Bitcoin transactions by injecting them into the simulator at some fixed rates. We generate 4000 validator nodes and randomly distribute them into shards. In the current Bitcoin setting, the block size limit is 1 MB, and the average size of a transaction is 500 bytes, hence, each block contains approximately 2000 transactions. We evaluate the system performance with 4, 8, 12, and 16 shards, which are the number of shards that were used in previous studies \cite{kokoris2018omniledger,zamani2018rapidchain,nguyen2019optchain}.

\subsubsection{Throughput} The experiment in this section illustrates how malicious transactions affect the system throughput. In order to find out the best throughput of the system, we gradually increase the transaction rate (i.e., the rate at which transactions are injected into the system) and observe the final throughput until the throughput stops increasing. At 16 shards, the best throughput is about 4000 tps, which is achieved when the transaction rate is about 5000 tps. For this experiment, we fix the rate at 5000 tps so that the system is always at its best throughput with respect to the number of shards.

To perform the attack, the attacker runs \Cref{algo:tx} to generate some portions of malicious transactions into shard 0. For example, if 10\% of transactions are malicious, then at each second, 500 transactions will be put into shard 0, and the rest 4,500 txs are distributed into shards according to their hash value. The results are shown in the \Cref{fig:thru-shard}.

\begin{figure}
	\centering
% 	\hspace{-20px}
	\subfloat[Impact on system throughput]{
	    \includegraphics[width=0.5\linewidth]{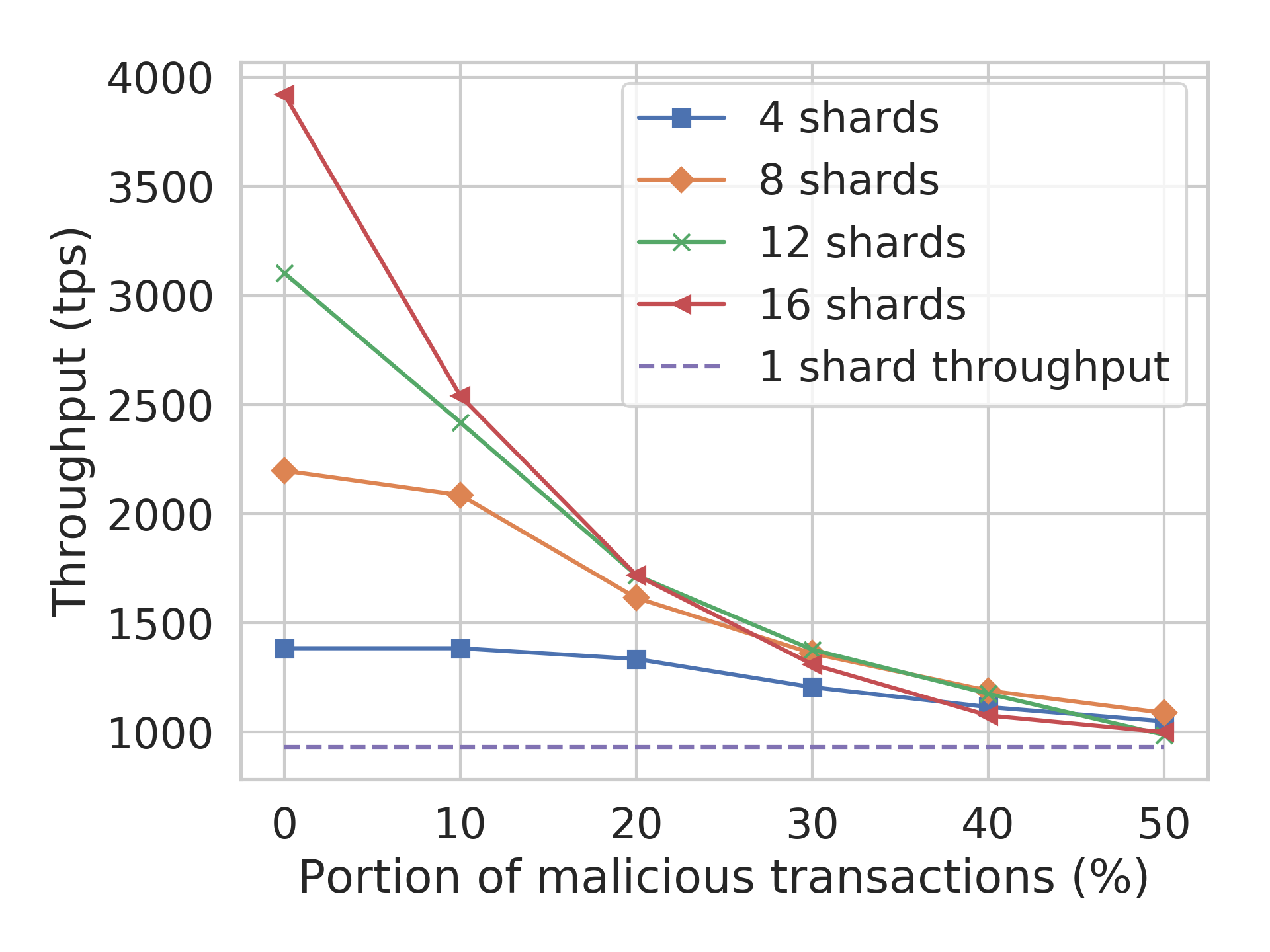}
    	\label{fig:thru-shard}
	}%
	\subfloat[Impact on system latency]{
	    \includegraphics[width=0.5\linewidth]{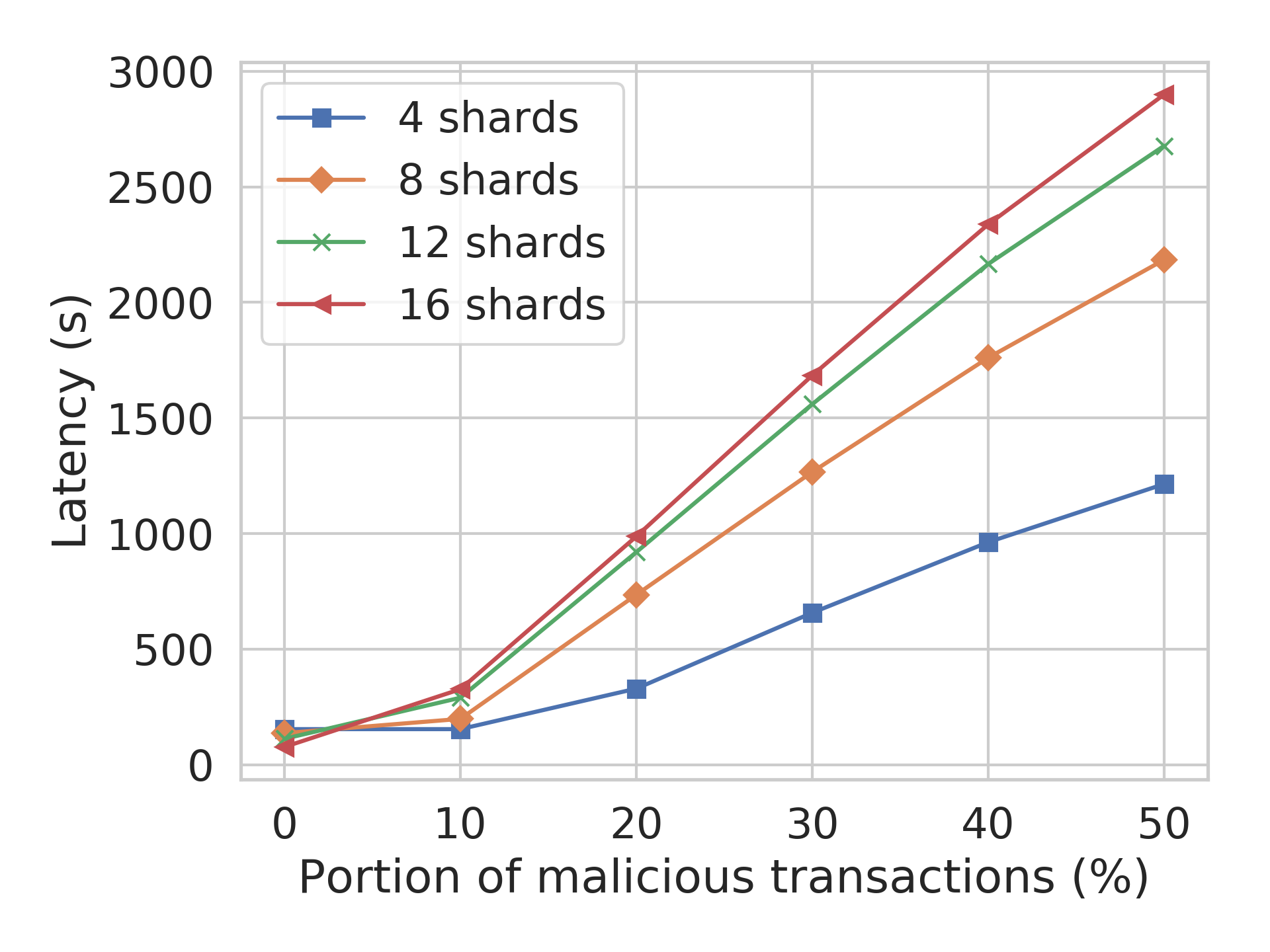}
	    \label{fig:lat-shard}
    }
	\caption{Impact on system throughput and latency}
	\label{fig:shard}
% 	\vspace{-3.5mm}
\end{figure}

% \begin{figure}
% 	\centering
% 	\includegraphics[width=0.8\linewidth]{thru-shard-10m.png}
% 	\caption{Impact on system throughput and latency}
% 	\label{fig:thru-shard}
% \end{figure}

At 0\%, the system is not under attack, the system achieves its best throughput with respect to the number of shards. The horizontal dashed line illustrates the throughput of 1 shard, which is the lower bound of the system throughput. As can be seen, when we increase the number of malicious transactions, the system throughput rapidly decreases. This behavior can be explained as we have multiple cross-shard transactions that are associated with the attacked shard, their delays could produce a severe cascading effect that ends up hampering the performance of other shards. Thus, the throughput as a whole is diminished. Moreover, we can observe that higher numbers of shards are more vulnerable to the attack. With 16 shards, the performance reduces exponentially as we increase the portion of malicious transactions to 50\%. Specifically, with only 20\% malicious transactions, the throughput was reduced by more than half. This behavior confirms our prior preliminary analysis. 

Another interesting observation is that at 50\% malicious transactions, the throughput nearly reaches its lower bound, hence, the sharding system would not be any faster than using only 1 shard. At this time, the attacker has accomplished its goal, that is, making the attacked shard the bottleneck of the whole system. 50\% malicious transactions translates to 2,500 malicious transactions that are sent to the system at each second, previous experiments and analysis in \Cref{tab:maltx} show that a normal laptop could easily generate more than 100,000 malicious transactions per second.

\subsubsection{Latency} 
%In this experiment, 
We analyze the impact of the single-shard flooding attack on the system latency, which is the average amount of time needed to confirm a transaction. To avoid backlog when the system is not under attack, for each number of shards, the transaction rate is set to match the best system throughput. The results are shown in the \Cref{fig:lat-shard}. In the same manner as the previous experiment, the attack effectively increases the latency of the system as a whole. We shall see in the next experiment that the attack creates a serious backlog in the mempool of the shards, thereby increasing the waiting time of transactions and eventually raising the average latency.

This experiment also corroborates the experiment on system throughput as greater numbers of shards are also more vulnerable to the attack. With 16 shards, although it provides the fastest transaction processing when the system is not under attack, nevertheless, it becomes the slowest one even with only 10\% of malicious transactions. Additionally, when the attacker generates 20\% malicious transactions, the latency is increased by more than 10 times. Therefore, we can conclude that although adding more shards would help improve the system performance, under our attack, the system would become more vulnerable.

% \begin{figure}
% 	\centering
% 	\includegraphics[width=0.7\linewidth]{lat-shard-10m.png}
% 	\caption{Impact on system latency}
% 	\label{fig:lat-shard}
% \end{figure}

\subsubsection{Queue/Mempool size}
In the following experiments, we investigate the impact of this single-shard flooding attack on the queue (or mempool) size of shard 0, which is the shard that is under attack. This gives us insights into how malicious transactions cause a backlog in the shard. Firstly, we fix the number of shards and vary the portion of malicious transactions. \Cref{fig:queue-ratio} illustrates the queue size over time of shard 0 with a system of 16 shards where each line represents the portion of malicious transactions. When the system is not under attack, the queue size is stable with less than 15,000 transactions at any point in time. As we put in only 10\% malicious transactions, the queue size reaches more than 2 million transactions. 

Note that under our attack, the transactions are injected into shard 0 as a rate that is much higher than its throughput, thus, the queue will keep on increasing until all transactions have been injected. At this point, transactions are no longer added to the shard and the shard is still processing transactions from the queue, hence, the queue size decreases. This explains why the lines (i.e., queue size) go down towards the end of the simulation.

The result also demonstrates that the congestion gets worse as we increase the malicious transactions. Due to the extreme backlog, transactions have to wait in the mempool for a significant amount of time, thereby increasing their waiting time. This explains the negative impact of malicious transactions on system throughput and latency.

\begin{figure}
	\centering
% 	\hspace{-20px}
	\subfloat[Number of shards = 16]{
	    \includegraphics[width=0.5\linewidth]{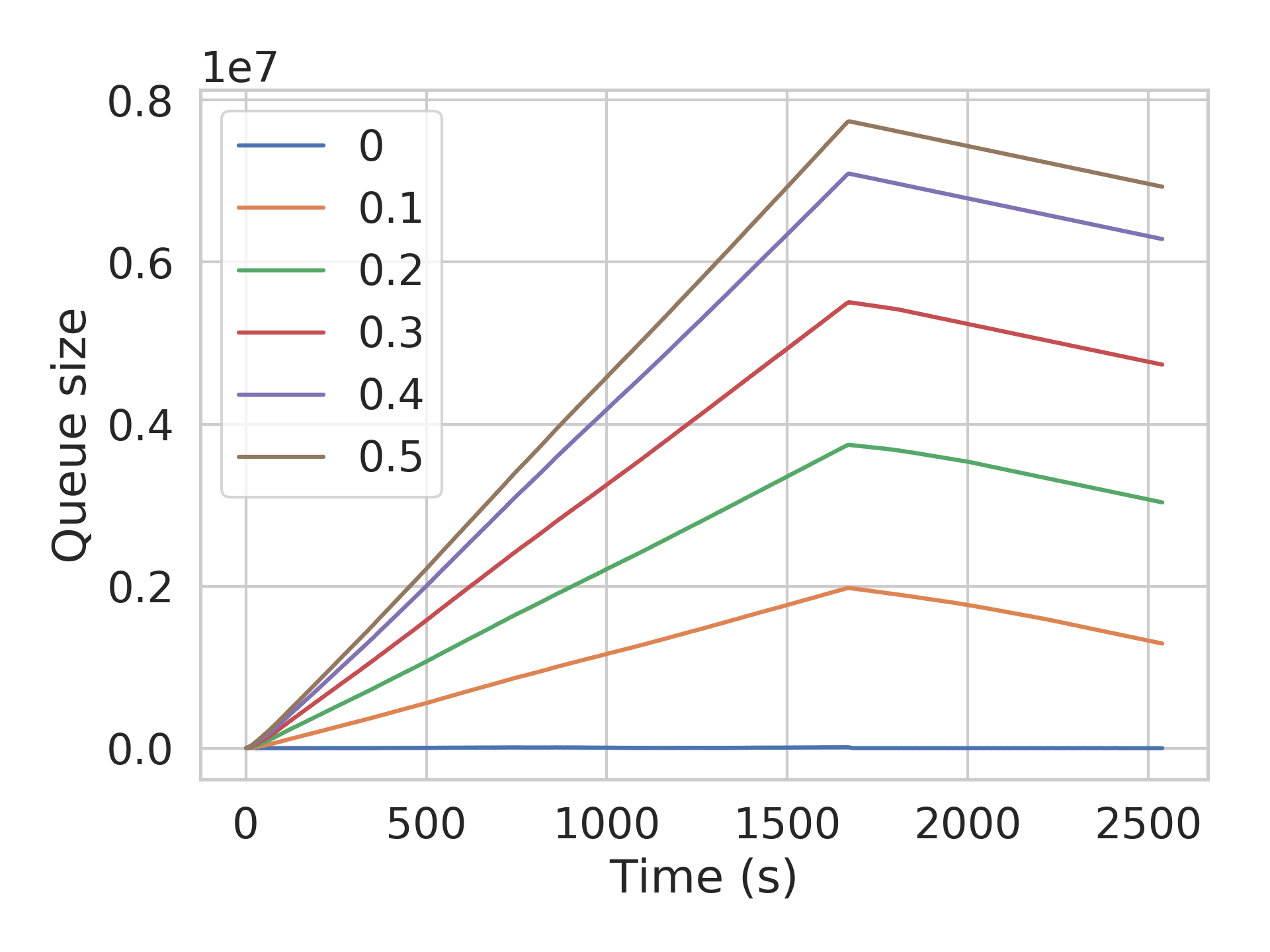}
    	\label{fig:queue-ratio}
	}%
	\subfloat[20\% malicious transactions]{
	    \includegraphics[width=0.5\linewidth]{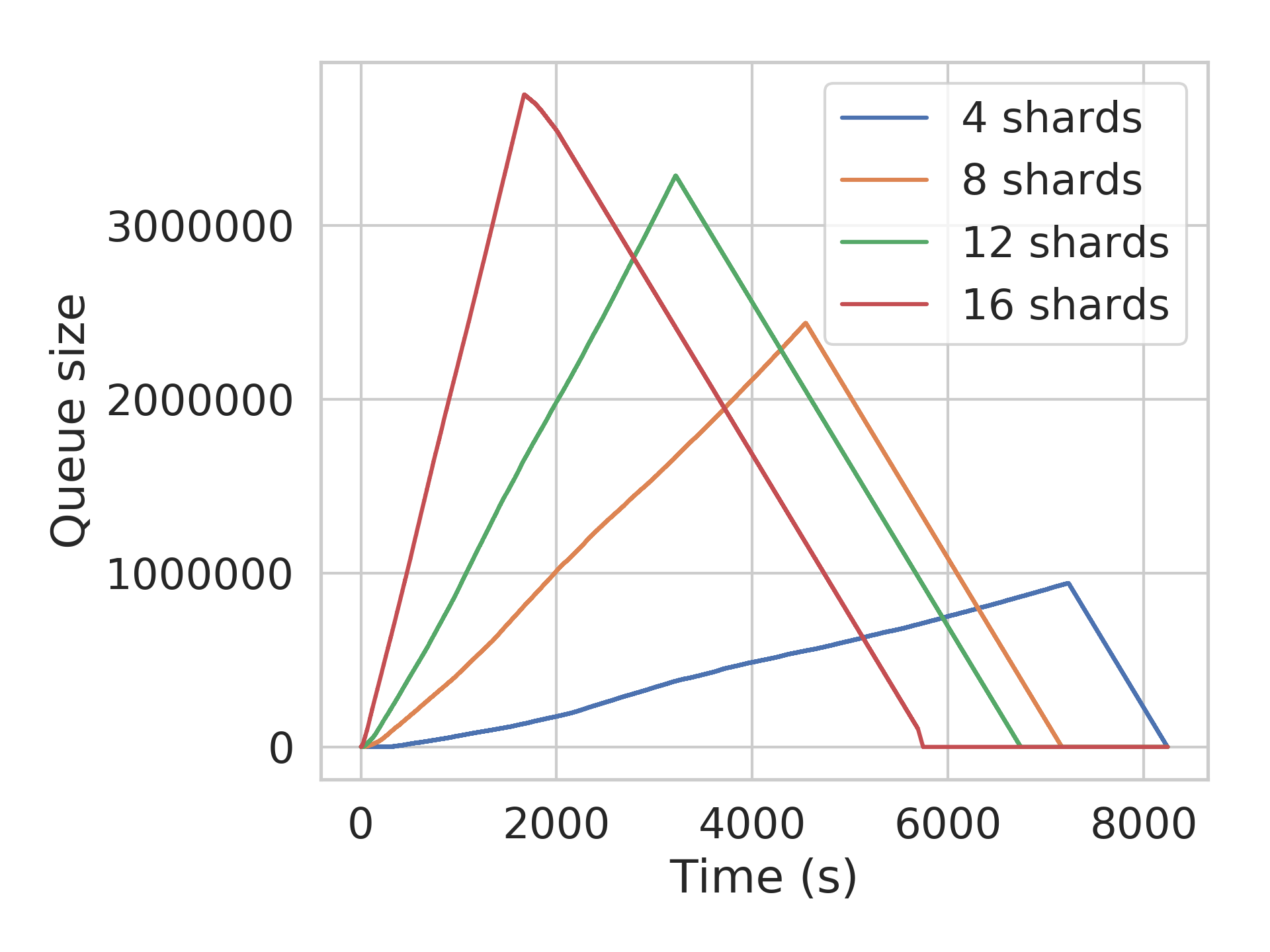}
	    \label{fig:queue-shard}
    }
	\caption{Impact on the attacked shard's queue size}
	\label{fig:queue}
% 	\vspace{-3.5mm}
\end{figure}

% \begin{figure}
% 	\centering
% 	\includegraphics[width=0.8\linewidth]{queue-ratio.png}
% 	\caption{Impact on the attacked shard's queue size at 16 shards}
% 	\label{fig:queue-ratio}
% \end{figure}

%\mt{20\% malicious transations: Is it too much?}

Next, we observe the queue size with different numbers of shards. \Cref{fig:queue-shard} presents the impact of the attack with 20\% malicious transactions at different numbers of shards. As can be seen, when we increase the number of shards, the backlog of transactions builds up much faster and greater due to the fact that we are having more cross-shard transactions. This result conforms to our previous claim that greater numbers of shards are more vulnerable to the attack.

\subsubsection{Summary} The experiments presented in this section have shown that our attack effectively reduce the performance of the whole system by attacking only a single shard. By generating malicious transactions according to \Cref{algo:tx}, the attacker easily achieves its goal of limiting the system performance to the throughput of one shard. Our preliminary analysis in \Cref{ssec:analysis} shows that the attacker is totally capable of generating an excessive amount of malicious transactions at low cost, thereby demonstrating the practicality of the attack.

\section{Countermeasure}\label{sec:countermeasure}
As we have argued that using hash values to determine the output shards is susceptible to the single-shard flooding attack, we delegate the task of determining the output shards without using hash values to the validators. To achieve that task, we consider the validators running a deterministic transactions sharding algorithm. The program takes the form of $S_{out} = txsharding(tx, st)$ in which it ingests as inputs a blockchain state $st$ and the transaction $tx$, and generates the output shard ID $S_{out}$ of $tx$ calculated at state $st$. Moreover, the algorithm $txsharding$ is made public. The requirement for $txsharding$ is that it is \textit{deterministic} and has a \textit{load-balancing} mechanism to balance the load among the shards. Finally, we assume that $txsharding$ does not use a transaction's hash as the basis for determining its output shard.

A simple approach to implement $txsharding$ is to devise an algorithm in which, for each input $tx$, $S_{out}$ is determined by choosing the shard that is currently having the least number of transactions. In other words, $txsharding$ always distributes the transactions into shards evenly, thereby balancing the load among the shards. For an optimal $txsharding$, Nguyen et al. \cite{nguyen2019optchain} proposes OptChain, a transactions placement algorithm that can both balance the load and minimize the number of cross-shard transactions. A technical overview of Optchain is given in Appendix B. With the load-balancing mechanism of $txsharding$, the attacker can no longer flood any one shard with a superfluous amount of transactions due to the fact that the transactions are always distributed evenly into shards. 

However, the main challenge is how to run $txsharding$ so that we can prevent adversaries from manipulating its output. By the nature of blockchain, the validators are untrusted, a straw-man approach is to run $txsharding$ on-chain and let the validators reach consensus on the output of $txsharding$. Hence, the validators would have to reach consensus on every single transaction. Nonetheless, this alone dismisses the original idea of sharding, which is to improve blockchain by parallelizing the consensus work and storage, i.e., each validator only handles a disjoint subset of transactions. To avoid costly on-chain consensus on the output of $txshading$, we need to execute the algorithm off-chain while ensuring that the operation is tamper-proof in the presence of malicious validators. To tackle this challenge, in this work, we leverage the Trusted Execution Environment (TEE) to securely execute the transaction sharding algorithm on the validators.

TEE in a computer system is realized as a module that performs some verifiable executions in such a way that no other applications, even the OS, can interfere. Simply speaking, a TEE module is a trusted component within an untrusted system. In this work, we consider using TEE that supports issuing remote attestations proving the integrity of the software running inside the TEE module. 
% These signatures are signed by the private keys that are known only by the TEE hardware. 
Intel SGX \cite{intel2014software} is one of the most commonly used implementations of TEE in which remote attestation is well-supported. However, TEE does not offer satisfactory availability guarantees as the hardware could be arbitrarily terminated by a malicious host or simply by losing power.

In the proposed countermeasure, validators are equipped with TEE modules that assist clients in determining the transactions' output shard with an attestation to prove the correctness of execution (most modern Intel CPUs from 2014 support Intel SGX). When a client issues a transaction to a validator, the validator will run its TEE module to get the output shard of that transaction, together with an \textit{attestation} to prove the code's integrity and the correctness of the execution. The attestation is a digital signature generated by the TEE's private key (Section \ref{ssec:func}). The client can verify the computation using the attestation and then send the transaction with the attestation to the system in the same manner as issuing an ordinary transaction. With this concept, we can rest assured that an attacker cannot manipulate transactions to overwhelm a single shard. Additionally, we do not need the whole blockchain validators to reach consensus on a transaction's output shard, this computation is instead done off-chain by one or some small amount of validators. Its realization, however, encounters some challenges when using TEE in an untrusted network:

\begin{itemize}
    \item A malicious validator can terminate the TEE at its discretion, which results in losing its state. The TEE module must be designed to tolerate such failure.
    \item Although the computation inside the TEE is trusted and verifiable via attestation, a malicious validator can deceive the TEE module by feeding it with fraudulent data.
\end{itemize}

To overcome these challenges, we aim to design a \textit{stateless} TEE module where any persistent state is stored in the blockchain. To obtain the state from the blockchain, the TEE module acts as a blockchain client to query the block headers from the blockchain, thereby ensuring the correctness of the data (this is how we can exploit the immutability of blockchain to overcome pitfalls of TEE modules). With this design, even when some TEE modules are arbitrarily shut down, the security properties of the protocol are not affected.

\subsection{System Overview and Security Goals}
In this section, we present an overview of our system for the countermeasure and establish some security goals.

% Our countermeasure realizes a secure execution of transactions sharding algorithm, which is a deterministic program. We consider  

\subsubsection{System overview}
Our system considers two types of entities: clients and validators
\begin{itemize}
    \item \textit{Clients} are the end-users of the system who are responsible for generating transactions. The clients are not required to be equipped with a TEE-enabled platform. In fact, the clients in our system are extremely lightweight.
    \item \textit{Validators} in each shard maintain a distributed append-only ledger, i.e. a blockchain, of that shard with an intra-shard consensus protocol. Validators require a TEE-enabled platform to run the transaction sharding algorithm.
\end{itemize}

For simplicity, we assume that a client has a list of TEE-enabled validators and it can send requests to multiple validators to tolerate certain failures. Each TEE-enabled validator has $txsharding$ installed in its TEE module.  We also assume that the TEE module in each of the validators constantly monitors the blockchain from each shard, so that it always has the latest state of the shards.

\begin{figure}
    \centering
    \includegraphics[width=\linewidth]{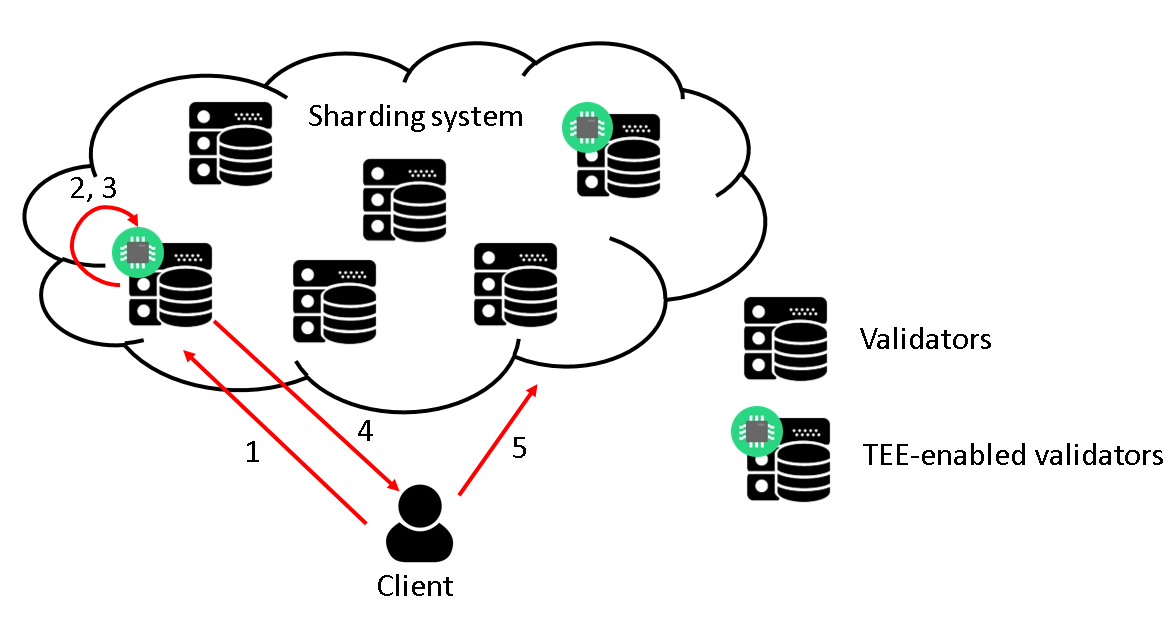}
    \caption{System overview}
    \label{fig:workflow}
\end{figure}

Denoting $ENC_k(m)$ as the encryption of message $m$ under key $k$ and $DEC_k(c)$ as the decryption of ciphertext $c$ under key $k$, the steps for computing the output shard for a transaction is as follows (\Cref{fig:workflow}):
\begin{enumerate}
    \item Client $C$ sends the transaction $tx$ to a TEE-enabled validators.  $C$ obtains the public key $pk_{TEE}$ of a validator's TEE, computes $inp = ENC_{pk_{TEE}}(tx)$, and sends $inp$ to the validator.
    \item The validator loads $inp$ into its TEE module and starts the execution of $txsharding$ in TEE.
    \item The TEE decrypts the $inp$ using its private key, and executes the $txsharding$ using $tx$ and the current state $st$ of the blockchain. Then, the output $S_{out}$ is generated together with the state $st$ upon which $S_{out}$ is determined, and a signature $\sigma_{TEE}$ proving the correct execution.
    \item The validator then send $(S_{out}, st, \sigma_{TEE}, h_{tx})$ to $C$ where $h_{tx}$ is the hash of the transaction. $C$ verifies $\sigma_{TEE}$ before sending $(\sigma_{TEE},tx)$ to the blockchain network for final validation. If $C$ sends a request to more than one validator, $C$ would choose the $S_{out}$ that reflects the latest state. 
    
    Note that $C$ could choose an outdated $S_{out}$, however, other entities can validate if a pair $(S_{out}, st)$ is indeed the output of a TEE. The blockchain system can simply reject transactions whose $S_{out}$ was computed based on an outdated $st$

    \item Upon receiving $(\sigma_{TEE},tx)$, the validators again verify $\sigma_{TEE}$ before proceeding with relaying and processing the transaction.
\end{enumerate}

% In practice, TEE platforms like Intel SGX performs the remote attestation as follows. The attestation for a correct computation takes the form of a signature $\pi$ from the output of TEE. Suppose Intel SGX is the implementation of TEE and the execution on TEE results in an output and an attestation $\sigma_{TEE}$, as indicated in \cite{cheng2019ekiden}, the validator sends $\sigma_{TEE}$ to the Intel Attestation Service (IAS) provided by Intel. Then IAS verifies  $\sigma_{TEE}$ and replies with $\pi = (b, \sigma_{TEE}, \sigma_{IAS})$, where $b$ indicates whether $\sigma_{TEE}$ is valid or not, and $\sigma_{IAS}$ is a signature over $b$ and $\sigma_{TEE}$ by the IAS. Since $\pi$ is basically a signature, it can be verified without using TEE or having to contact the IAS.

\subsubsection{Adversarial model and Security goals}\label{sec:countermeasurethreat}
In the threat model in \Cref{ssec:threat}, the attacker only plays the role of a client, however, we stress that the countermeasure must not violate the adversarial model of blockchain, which is working with malicious validators. Thus, in designing the countermeasure system, we extend the previous threat model as follows.

In the same manner as previous work on TEE-enabled blockchain \cite{lind2019teechain,cheng2019ekiden}, we consider an adversary who controls the operating system and any other high-privilege software on the validators. Attackers may drop, interfere, or send arbitrary messages at any time during execution. We assume that the adversary cannot break the hardware security enforcement of TEE. The adversary cannot access processor-specific keys (e.g., attestation and sealing key) and it cannot access TEE's memory that is encrypted and integrity-protected by the CPU. 

The adversary can also corrupt an arbitrary number of clients. Clients are lightweight, they only send requests to the validator the get the output shard of a transaction. They can verify the computation without TEE. We assume honest clients trust their platforms and software, but not that of others. We consider that the blockchain will perform prescribed computation correctly and is always available.

With respect to the adversarial model, we define the security notions of interest as follows:
\begin{enumerate}
    \item Correct execution: the output of a TEE module must reflect the correct execution of $txsharding$ with respect to inputs $tx$ and $st$, despite malicious host.
    \item The system is secure against the aforementioned single-shard flooding attack.
    \item Stateless TEE: the TEE module does not need to retain information regarding previous states or computations.
\end{enumerate}

\subsubsection{Blockchain sharding configuration}
For ease of presentation, we assume a sharding system that resembles the OmniLedger blockchain \cite{kokoris2018omniledger}. Suppose the sharding system has $n$ shards, for each $i \in \{1, 2, ..., n\}$, we denote $BC_i$ and $BH_i$ as the whole ledger and block headers of shard $S_i$, respectively. Furthermore, we assume that each shard $S_i$ keeps track of its UTXO database, denoted by $U_i$. Given a shard $S_i$, each validator in $S_i$ monitors the following database: $BC_i$, $BH_{j\neq i}$, and $U_{1,2, ..., n}$.

As we use an Omniledger-like blockchain, each shard elects a leader who is responsible for accepting new transactions to the shard. For simplicity, we consider that the sharding system provides an API $validate(tx)$ that takes a transaction $tx$ as the input and performs the transaction validation mechanism on $tx$. $validate(tx)$ returns true if $tx$ is successfully committed to the blockchain, otherwise, it returns false.

Additionally, we consider the system uses a signature scheme $\Sigma(G, Sig, Vf)$ that is assumed to be EU-CMA secure (Existential Unforgeability under a Chosen Message Attack). ECDSA is a suitable signature scheme in practice \cite{johnson2001elliptic}. Moreover, the hash function $\mathcal{H}(\cdot)$ used by the system is also assumed to be collision resistant: there exists no efficient algorithm that can find two inputs $a \neq b$ such that $\mathcal{H}(a) = \mathcal{H}(b)$.

Finally, we assume that each TEE generates a public/secret key pair and the public key is publicly available to all entities in the network. In practice, the public keys could be stored in a global identity blockchain.

% The main objective of the protocol is to make sure that the Client can securely and verifiably determine the output shard of a transaction.

% We assume the hash function $\mathcal{H}(\cdot)$ is collision resistant: there exists no efficient algorithm that can find two inputs $a \neq b$ such that $\mathcal{H}(a) = \mathcal{H}(b)$.

\subsection{Modeling Functionality of Blockchain and TEE} \label{ssec:func}
We specify the ideal blockchain $\mathcal{F}_{blockchain}$ as an appended decentralized database as in \cref{fig:fb}. $\mathcal{F}_{blockchain}$ stores $\mathbf{DB} = \{DB_i | i \in \{1, 2, ..., n\}\}$ which represents a set of blockchains that are held by shard 1, 2, ..., $n$, respectively. Each blockchain $DB_i$ is indexed by the transactions' hash value $h_{tx}$. We assume that by writing to the blockchain of a shard, all validators of the shard reach consensus on that operation.

\begin{figure}
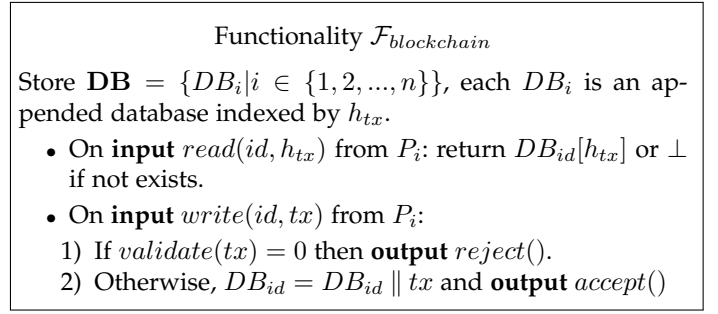

	\noindent\fbox{%
		\parbox{\columnwidth}{%
			\begin{center}
				Functionality $\mathcal{F}_{blockchain}$
			\end{center}
			
			Store $\mathbf{DB} = \{DB_i | i \in \{1, 2, ..., n\}\}$, each $DB_i$ is an appended database indexed by $h_{tx}$.
			
			\begin{itemize}
			    \item On \textbf{input} $read(id, h_{tx})$ from $P_i$: return $DB_{id}[h_{tx}]$ or $\bot$ if not exists.
			\end{itemize}
			\begin{itemize}
				\item On \textbf{input} $write(id, tx)$ from $P_i$:
				\begin{enumerate}
					\item If $validate(tx) = 0$ then \textbf{output} $reject()$.
					\item Otherwise, $DB_{id} = DB_{id} \concat tx$ and \textbf{output} $accept()$
				\end{enumerate}
			\end{itemize}
		}
	}
	\caption{Ideal blockchain $\mathcal{F}_{blockchain}$}
	\label{fig:fb}
\end{figure}

We specify the ideal TEE $\mathcal{F}_{TEE}$ that models a TEE module in \Cref{fig:ftee}, following the formal abstraction in \cite{pass2017formal}. On startup, $\mathcal{F}_{TEE}$ generates a public/secret key pair $(pk_{TEE}, sk_{TEE})$. Then, it publishes $pk_{TEE}$ together with a remote attestation proof proving the integrity of the code in TEE, including the key generation code. With the public key, other entities in the network are able to verify messages signed by $\mathcal{F}_{TEE}$'s secret key. Since $pk_{TEE}$ is bound to the TEE code, signatures under $sk_{TEE}$ can be used to prove the integrity of the execution output. Therefore, they can be used as attestations for the correct execution of TEE.

In practice, TEE platforms like Intel SGX perform the remote attestation as follows. Suppose the execution on TEE results in an output and an attestation $\sigma_{TEE}$, as indicated in \cite{cheng2019ekiden}, the validator sends $\sigma_{TEE}$ to the Intel Attestation Service (IAS) provided by Intel. Then IAS verifies  $\sigma_{TEE}$ and replies with $\pi = (b, \sigma_{TEE}, \sigma_{IAS})$, where $b$ indicates whether $\sigma_{TEE}$ is valid or not, and $\sigma_{IAS}$ is a signature over $b$ and $\sigma_{TEE}$ by the IAS. Since $\pi$ is basically a signature, it can be verified without using TEE or having to contact the IAS.

A TEE module is an isolated software container that is installed with a program that, in this work, is a transaction sharding algorithm. $\mathcal{F}_{TEE}$ abstracts a TEE module as a trusted third party for confidentiality, execution, and authenticity with respect to any entities that is a part of the system. $prog$ is a program that is installed to run in a TEE module; the input and output of $prog$ are denoted by $inp$ and $outp$, respectively.

\begin{figure}
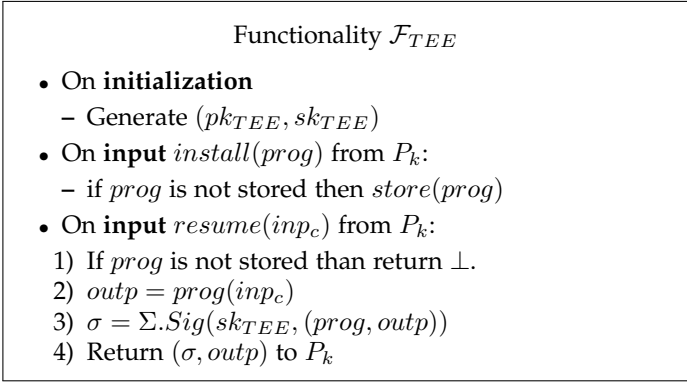

	\noindent\fbox{%
		\parbox{\columnwidth}{%
			\begin{center}
				Functionality $\mathcal{F}_{TEE}$
			\end{center}
			
			\begin{itemize}
			    \item On \textbf{initialization}
			    \begin{itemize}
			        \item Generate $(pk_{TEE}, sk_{TEE})$
			    \end{itemize}
			\end{itemize}
			
			\begin{itemize}
			    \item On \textbf{input} $install(prog)$ from $P_k$:
			    \begin{itemize}
			        \item if $prog$ is not stored then $store(prog)$
			    \end{itemize}
			\end{itemize}
			\begin{itemize}
				\item On \textbf{input} $resume(inp_c)$ from $P_k$:
				\begin{enumerate}
					\item If $prog$ is not stored than return $\bot$.
				    \item $outp = prog(inp_c)$
				    \item $\sigma = \Sigma.Sig(sk_{TEE}, (prog, outp))$
				    \item Return $(\sigma, outp)$ to $P_k$
				\end{enumerate}
			\end{itemize}
		}
	}
	\caption{Ideal TEE $\mathcal{F}_{TEE}$}
	\label{fig:ftee}
\end{figure}

\begin{figure}
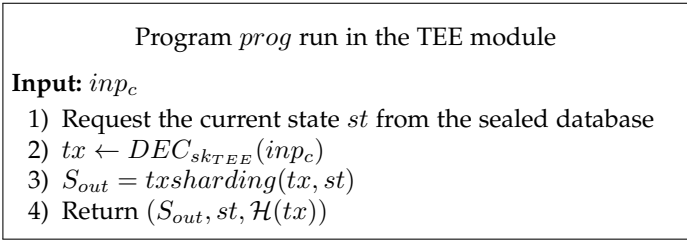

	\noindent\fbox{%
		\parbox{\columnwidth}{%
			\begin{center}
				Program $prog$ run in the TEE module
			\end{center}
			
			\textbf{Input:} $inp_c$
			
			\begin{enumerate}
			    \item Request the current state $st$ from the sealed database
			    \item $tx \gets DEC_{sk_{TEE}}(inp_c)$
			    \item $S_{out} = txsharding(tx, st)$
			    \item Return $(S_{out}, st, \mathcal{H}(tx))$
			\end{enumerate}
		}
	}
	\caption{Program $prog$ run in the TEE module}
	\label{fig:prog}
\end{figure}

On initialization, the TEE module needs to download and monitor $U_i$ for $i \in {1,2,..,n}$. These data are encrypted using the TEE's secret key and then stored in the host storage, which is also referred to as sealing. In this way, the TEE will make sure that its data on the secondary storage cannot be tampered with by a malicious host. To ensure that the TEE always uses the latest version of the sealed UTXO database, rollback-protection systems such as ROTE \cite{matetic2017rote} can be used.

\Cref{fig:prog} defines the program $prog$ that is installed in the TEE module to be used in this work. As can be seen, the program decrypts the encrypted input $inp_c$ using the secret key $sk_{TEE}$. $txsharding$ is implemented inside $prog$ to securely execute the transaction sharding algorithm. The program returns the output shard $S_{out}$, the state $st$ upon which $S_{out}$ was computed, and the hash of the transaction $h_{tx}$. This hash value is used to prevent malicious hosts from feeding the TEE module with fake transactions, which will be discussed in more detail in the next subsection.

Upon running $prog$ with the input $inp_c$, the TEE module obtains the signature $\sigma_{TEE}$ over the output of $prog$ and the code of $prog$ using its private key. Finally, the TEE module returns to the host validator $\sigma_{TEE}$, and the output $outp$ from $prog$.

\subsection{Formal Specification of the Protocol}
Our proposed system supports two main APIs for the end-users: (1) $newtx(tx)$ handles the secure computation of a transaction $tx$, and (2) $read(id, h_{tx})$ returns the transaction that has the hash value $h_{tx}$ from shard $id$.

The protocol for validators is formally defined in \cref{fig:prot}, which relies on $\mathcal{F}_{TEE}$ and $\mathcal{F}_{blockchain}$. The validator accepts two function calls from the clients: $request(tx_c)$ and $process(S_{out}, \sigma_{TEE}, tx)$. $request(tx_c)$ takes as input a transaction that is encrypted by the public key of TEE and sends $tx_c$ to the TEE module. For simplicity, we assume that the validator is TEE-enabled, if not, the validator simply discards the $request(tx_c)$ function call. Since $tx_c$ is encrypted by $pk_{TEE}$, a malicious host cannot tamper with the transaction. The validator waits until the TEE returns an output and relays that output to the function's caller. 

Note that as the output of the TEE includes the transaction's hash $h_{tx}$, the client can check that the TEE indeed processed the correct transaction $tx$ originated from the client. This is possible because of the end-to-end encryption of $tx$ between the client and the TEE. Furthermore, since $\sigma_{TEE}$ protects the integrity of $S_{out}$, the client can verify that $S_{out}$ was not modified by a malicious validator.

The function $process(S_{out}, \sigma_{TEE}, tx)$ receives as input the transaction $tx$, $\sigma_{TEE}$, and output shard $S_{out}$ of $tx$. The validator also verifies $\sigma_{TEE}$ before making a call to $\mathcal{F}_{blockchain}$ to start the transaction validation for $tx$.

% $\mathcal{F}_{TEE}$ computes $S_{out}$ with the current state $st$ and send both to the host $P_i$. $P_i$ forwards the result to the calling client. $\sigma_{TEE}$ protects the integrity of $S_{out}$ and $st$, thus a malicious validator or client cannot tamper with it.

\begin{figure}
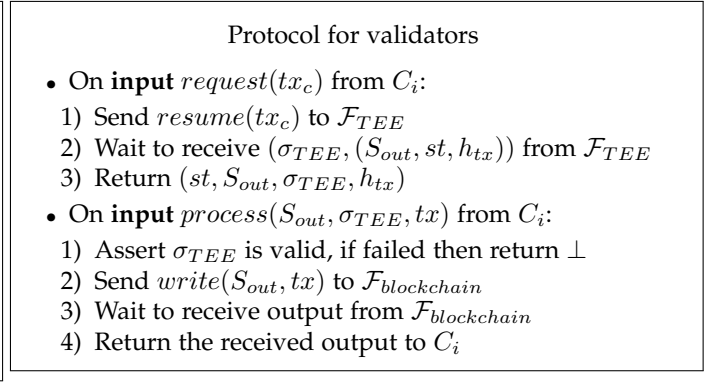

	\noindent\fbox{%
		\parbox{\columnwidth}{%
			\begin{center}
				Protocol for validators
			\end{center}
			
			\begin{itemize}
				\item On \textbf{input} $request(tx_c)$ from $C_i$:
				\begin{enumerate}
					\item Send $resume(tx_c)$ to $\mathcal{F}_{TEE}$
					\item Wait to receive $(\sigma_{TEE}, (S_{out}, st, h_{tx}))$ from $\mathcal{F}_{TEE}$
				% 	\item $\sigma = (\sigma_{TEE}, pk_{TEE})$
				    % \item Attest and receive $\pi = (b, \sigma_{TEE}, \sigma_{IAS})$
					\item Return $(st, S_{out},  \sigma_{TEE}, h_{tx})$
				\end{enumerate}
			\end{itemize}
			
			\begin{itemize}
				\item On \textbf{input} $process(S_{out}, \sigma_{TEE}, tx)$ from $C_i$:
				\begin{enumerate}
				    \item Assert $\sigma_{TEE}$ is valid, if failed then return $\bot$
					\item Send $write(S_{out}, tx)$ to $\mathcal{F}_{blockchain}$
					\item Wait to receive output from $\mathcal{F}_{blockchain}$
					\item Return the received output to $C_i$
				\end{enumerate}
			\end{itemize}
% 			\begin{itemize}
% 				\item On \textbf{input} $claim(st, S_{out}, \sigma, epk_i, tx)$ from $C_i$:
% 				\begin{enumerate}
% 					\item Send $write(tx, (S_{out}, \sigma)$ to $\mathcal{F}_{blockchain}$
% 					\item if receive $reject()$ then return $\bot$
% 					\item Otherwise, send $exe(tx)$ to $\mathcal{F}_{TEE}$
% 					\item Receive $(\sigma_{TEE}, S'_{out}, st)$ from $\mathcal{F}_{TEE}$
% 					\item if $S'_{out} = S_{out}$ then return $S_{out}$, otherwise, return $\bot$
% 				\end{enumerate}
% 			\end{itemize}
		}
	}
	\caption{Protocol for validators}
	\label{fig:prot}
\end{figure}

\Cref{fig:client} illustrates the protocol for the clients. To determine the output shard of a transaction $tx$, a client invokes the API $newtx(tx)$. First, to ensure the integrity of $tx$, the client encrypts $tx$ using the $pk_{TEE}$ and sends $tx_c$ to a TEE-enabled validator $P_k$. Upon receiving $(st, S_{out}, \sigma_{TEE}, h_{tx})$ from $P_k$, the client checks if the hash of $tx$ is equal to $h_{tx}$. This prevents a malicious validator from feeding a fake transaction to the TEE module to manipulate $S_{out}$. The client also verifies if the attestation $\sigma_{TEE}$ is correct. Afterward, the client sends the transaction together with $S_{out}$ and the attestation to the validators of the transaction's input and output shards for final validation. $newtx(tx)$ finally outputs any data received from the validators. The API $read(id, h_{tx})$ can be called when the client wants to obtain the transaction information from the blockchain. The function also returns any data received from $\mathcal{F}_{blockchain}$.

\begin{figure}
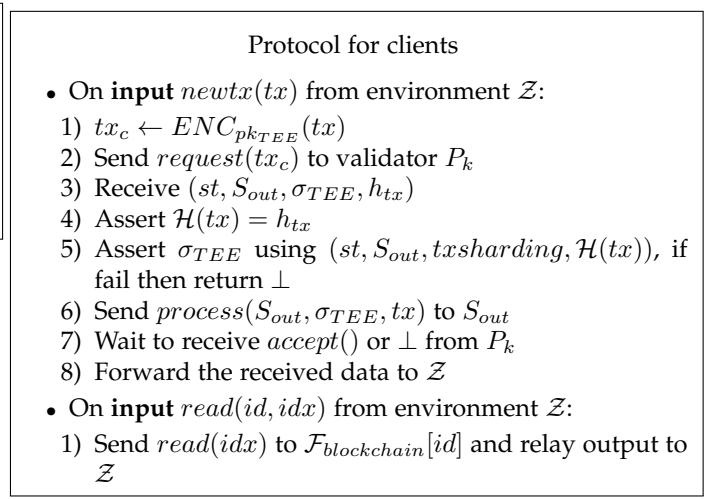

	\noindent\fbox{%
		\parbox{\columnwidth}{%
			\begin{center}
				Protocol for clients
			\end{center}
			
			\begin{itemize}
				\item On \textbf{input} $newtx(tx)$ from environment $\mathcal{Z}$:
				\begin{enumerate}
					\item $tx_c \gets ENC_{pk_{TEE}}(tx)$
					\item Send $request(tx_c)$ to validator $P_k$
					\item Receive $(st, S_{out}, \sigma_{TEE}, h_{tx})$
					\item Assert $\mathcal{H}(tx) = h_{tx}$
					\item Assert $\sigma_{TEE}$ using $(st, S_{out}, txsharding, \mathcal{H}(tx))$, if fail then return $\bot$
					\item Send $process(S_{out}, \sigma_{TEE}, tx)$ to $S_{out}$
					\item Wait to receive $accept()$ or $\bot$ from $P_k$
				    \item Forward the received data to $\mathcal{Z}$
				\end{enumerate}
			\end{itemize}
			
			\begin{itemize}
			    \item On \textbf{input} $read(id, idx)$ from environment $\mathcal{Z}$:
			    \begin{enumerate}
			        \item Send $read(idx)$ to $\mathcal{F}_{blockchain}[id]$ and relay output to $\mathcal{Z}$
			    \end{enumerate}
			\end{itemize}
		}
	}
	\caption{Protocol for clients}
	\label{fig:client}
\end{figure}

\section{Security Analysis and Performance Evaluation} \label{sec:sec-eval}
This section presents a detailed security analysis of the proposed countermeasure under the UC-model and evaluates the performance of the proof-of-concept implementation.

\vspace{-2pt}

\subsection{Security Analysis}
We first prove that the proposed protocol (1) only requires a stateless TEE and (2) is secure against the single-shard flooding attack, and then prove the correct execution security. By design, the $txsharding$ installed in the TEE depends solely on the transaction $tx$, and the state $st$ of the blockchain for computation. As $tx$ is the input and $st$ can be obtained by querying from the UTXO database, the TEE does not need to keep any previous states and computation, thus, it is stateless.

When the system makes decision on the output shard of a transaction, it relies on the $txsharding$ program which is assumed to not base its calculation on transactions' hash value. Therefore, as $txsharding$ also balances the load among the shards, no attackers can manipulate transactions to overwhelm a single shard, hence, the countermeasure is secure against the single-shard flooding attack.

The correct execution security of our system is proven in the Universal Composability (UC) framework \cite{canetti2001universally}. We refer the readers to Appendix A for our proof.

\vspace{-2pt}

\subsection{Performance Evaluation}
This section presents our proof-of-concept implementation as well as some experiments to evaluate its performance. As our countermeasure is immune to the single-shard flooding attack, our goal is to evaluate the overhead of integrating this solution into sharding. We implement the proof-of-concept using Intel SGX which is available on most modern Intel CPUs. With SGX, each implementation of the TEE module is referred to as an \textit{enclave}. The proof-of-concept was developed on Linux machines in which we use the Linux Intel SGX SDK 2.1 for development. We implement and test the protocol for validators using a machine equipped with an Intel Core i7-6700, 16GB RAM, and an SSD drive.

As we want to demonstrate the practicality of the countermeasure, the focus of this evaluation is three-fold: processing time, communication cost, and storage. The processing time includes the time needed for the enclave to monitor the block headers as well as to determine the output shard of a transaction requested by a client. The communication cost represents the network overhead incurred by the interaction between clients and validators to determine the output shards. Additionally, we measure the amount of storage needed when running the enclave.

In our proof-of-concept implementation, we use OptChain \cite{nguyen2019optchain} as $txsharding$. As OptChain determines the output shard based on the transaction's inputs, when obtaining the state from the UTXO database, we only need to load those transaction's inputs from the database. Our proof-of-concept uses Bitcoin as the blockchain platform.%, and the enclave is connected to the Bitcoin mainnet.

\textbf{Processing time.}
We calculate the processing time for determining the output shard by invoking the enclave with 10 million Bitcoin transactions (encrypted with the TEE public key) and measure the time needed to receive output from the enclave. This latency includes  (1) decrypting the transaction, (2) obtaining the latest state from the UTXO database in the host storage and (3) running $txsharding$. We observe that the highest latency recorded is only about 214 ms and it also does not vary much when running with different transactions. Considering that the average latency of processing a transaction in sharding is about 10 seconds \cite{nguyen2019optchain}, our countermeasure only imposes an additional 0.2 seconds for determining the output shard.

For a detailed observation, we measure the latency separately for each stage. The running time of $txsharding$ is negligible as the highest running time recorded is about 0.13 ms when processing a 10-input transaction at 16 shards. Decrypting the transaction is about 0.579 ms when using 2048-bit RSA, and a query to the UTXO database to fetch the current state takes about 213.2 ms. Hence, fetching the state from the UTXO database dominates the running time due to the fact that the database is stored in the host storage.

Another processing time that we consider is the time needed for updating the UTXO database when a new block is added to the blockchain. With an average number of 2000 transactions per block, each update takes about 65.7 seconds. But this latency does not affect the performance of the system since we can set up two different enclaves running in parallel, one for running $txsharding$, and one for monitoring the UTXO database. Therefore, the time needed to run $txsharding$ is independent of updating the database.

\textbf{Communication cost.}
Our countermeasure imposes some communication overhead over hash-based transaction sharding since the client has to communicate with the validator to determine the output shard. Specifically, the overhead includes sending the encrypted transaction to the validator and receiving a response. The response comprises the state (represented as the block number), output shard ID, attestation, and a hash value of the transaction. The communication overhead sums up to about 601 bytes needed for the client to get the output shard of a transaction. If we consider a communication bandwidth of 10 Mbps, the transmission time would take less than half a millisecond.

\textbf{Storage.}
According to the proposed system for the countermeasure, the enclave needs to store the UTXO database in the host storage, which is essentially the storage of the validator. As of Feb 2020, the size of Bitcoin's UTXO set is about 3.67 GB, which means that an additional 3.67 GB is needed in the validator's storage. However, considering that the validator's storage is large enough to store the whole ledger, which is about 263 GB as of Feb 2020, the extra data trivially accounts for 1.4\%.

\textbf{Summary.}
By evaluating a proof-of-concept of our countermeasure, the result in this section shows that our system imposes negligible overhead compared to the hash-based transaction sharding, which is susceptible to the single-shard flooding attack. Particularly, by incurring insignificant processing time, communication cost, and storage, our proposed countermeasure demonstrates the practicality as it can be integrated into existing sharding solutions without affecting the system performance.

\section{Conclusion}\label{sec:con}
In this paper, we have identified a new attack in existing sharding solutions. Due to the use of hash-based transaction sharding, an attacker can manipulate the hash value of a transaction to conduct the single-shard flooding attack, which is essentially a DoS attack that can overwhelm one shard with an excessive amount of transactions. We have thoroughly investigated the attack with multiple analyses and experiments to illustrate its damage and practicality. Most importantly, our work has shown that by overwhelming a single shard, the attack creates a cascading effect that reduces the performance of the whole system.

We have also proposed a countermeasure based on TEE that efficiently eliminates the single-shard flooding attack. The security properties of the countermeasure have been proven in the UC framework. Finally, with a proof-of-concept implementation, we have demonstrated that our countermeasure imposes negligible overhead and can be integrated into existing sharding solutions.

\bibliographystyle{IEEEtran}
% argument is your BibTeX string definitions and bibliography database(s)
\bibliography{bibliography}

\clearpage

\appendices

\crefalias{section}{appendix}
\section{Proof of correct execution in the UC framework}\label{app:proof}
\begin{figure}
	\noindent\fbox{%
		\parbox{\columnwidth}{%
			\begin{center}
				$\mathcal{F}_{cm}$: Ideal functionality of the countermeasure
			\end{center}
			
			\begin{itemize}
			    \item On \textbf{Initialization}:
			    \begin{enumerate}
			        \item $DB_i = \emptyset, \forall i \in [N]$
			    \end{enumerate}
				\item On \textbf{input} $newtx(tx)$ from any party $P_i$:
				\begin{enumerate}
				    \item relay input to $\mathcal{A}$
				    \item Request current state $st$ from the UTXO database
				    \item $S_{out} = txsharding(tx, st)$
				    \item $DB_{S_{out}} = DB_{S_{out}} || tx$
				    \item Send delayed $accept()$ to $P_i$
				\end{enumerate}
				\item On \textbf{input} $read(id, h_{tx})$ from any party $P_i$:
				\begin{enumerate}
				    \item If $DB_{id}[h_{tx}]$ is not available then return $\bot$
				    \item Otherwise, return $DB_{id}[h_{tx}]$
				\end{enumerate}
			\end{itemize}
		}
	}
	\caption{Ideal functionality of the countermeasure}
	\label{fig:fcm}
\end{figure}

In the UC framework, a real world involves parties running the proposed protocol, namely $\Pi_{cm}$. On the other hand, an ideal world consists of parties that interact with an ideal functionality $\mathcal{F}_{cm}$, a trusted third party that implements the APIs of the proposed protocol, i.e., $newtx(\cdot)$ and $read(\cdot)$. \Cref{fig:fcm} shows the definition of $\mathcal{F}_{cm}$. Any adversary $\mathcal{A}$ in the real world is introduced in the ideal world by a simulator $\mathcal{S}$ with an adversary model defined in \Cref{sec:countermeasurethreat}.

To prove that the proposed protocol $\Pi_{cm}$ achieves correct execution security, we show that: (1) $\mathcal{F}_{cm}$ achieves the correct execution security in the ideal world; and (2) the real and ideal worlds are indistinguishable to an external environment $\mathcal{Z}$. This implies that any attack violating security goals in the real world is translatable to a corresponding attack in the ideal one. This proves that the real world protocol also achieves the correct execution security. 

In $\mathcal{F}_{cm}$, whenever a party triggers $newtx(tx)$ with a transaction $tx$, the execution of $txsharding$ is performed internally by the ideal functionality $\mathcal{F}_{cm}$ based on $tx$ and the current state $st$ to determine $S_{out}$. Since $\mathcal{F}_{cm}$ is trusted under UC, $txsharding$ is guaranteed to be correctly executed. Furthermore, $\mathcal{F}_{cm}$ also validates the transaction in the blockchain and returns only $accept()$ or $\bot$ to the party, an adversary does not have control over the output shard $S_{out}$, hence, the adversary cannot tamper with $S_{out}$. Therefore, $\mathcal{F}_{cm}$ achieves correct execution security.

Let $\mathcal{A}$ be an adversary against the proposed protocol. Per Canetti \cite{canetti2001universally}, we say that $\Pi_{cm}$ UC-realizes $\mathcal{F}_{cm}$ if there exists a simulator $\mathcal{S}$, such that any environment $\mathcal{Z}$ cannot distinguish between interacting with the adversary $\mathcal{A}$ and $\Pi_{cm}$ or with the simulator $\mathcal{S}$ and the ideal functionality $\mathcal{F}_{cm}$. By that definition, we prove the following theorem:

\begin{thm}
The protocol $\Pi_{cm}$ in the $(\mathcal{F}_{TEE}, \mathcal{F}_{blockchain})$ hybrid model UC-realizes the ideal functionality $\mathcal{F}_{cm}$.
\end{thm}

\begin{proof}
We prove the indistinguishability between the real and ideal worlds through a series of \textit{hybrid steps} as commonly done in previous work \cite{cheng2019ekiden,lind2019teechain}. These hybrid steps start at $H_0$ - the real-world execution of $\Pi_{cm}$ in the $(\mathcal{F}_{TEE}, \mathcal{F}_{blockchain})$ hybrid model, and finally becomes the execution in the ideal world. The indistinguishability is proven in each step.

Hybrid $H_0$ is the real-world execution where parties run $\Pi_{cm}$ in the $(\mathcal{F}_{TEE}, \mathcal{F}_{blockchain})$ hybrid model.

Hybrid $H_1$ behaves in the same manner as $H_0$, except that $\mathcal{S}$ emulates $\mathcal{F}_{TEE}$ and $\mathcal{F}_{blockchain}$. First, $\mathcal{S}$ generates a key pair $(pk_{TEE},sk_{TEE})$ and publishes $pk_{TEE}$. Whenever $\mathcal{A}$ interacts with $\mathcal{F}_{TEE}$, $\mathcal{S}$ records messages sent by $\mathcal{A}$ and emulates $\mathcal{F}_{TEE}$’s behavior. Likewise, $\mathcal{S}$ emulates $\mathcal{F}_{blockchain}$ by storing $\mathbf{DB}$ internally. As the view of $\mathcal{A}$ in $H_1$ is identically simulated in $H_0$, $\mathcal{Z}$ cannot distinguish between $H_1$ and the execution $H_0$.

Hybrid $H_2$ proceeds as $H_1$. However, every time $\mathcal{A}$ communicates with $\mathcal{F}_{blockchain}$, $\mathcal{S}$ identically emulates $\mathcal{F}_{blockchain}$’s behavior for $\mathcal{A}$. As the view of $\mathcal{A}$ in $H_2$ are simulated when interacting with the ledger, then environment $\mathcal{Z}$ cannot distinguish between $H_2$ and $H_1$.

Hybrid $H_3$ modifies $H_2$ as follows. When $\mathcal{A}$ triggers $\mathcal{F}_{TEE}$ with a message $install(prog)$, $\mathcal{S}$ records a tuple $(\sigma_{TEE}, outp)$ for all subsequent $resume(\cdot)$ calls, where $outp$ is the output of $prog$ and $\sigma_{TEE}$ is an attestation under $sk_{TEE}$ over $outp$ and $prog$. $\mathcal{S}$ keeps a set of all such tuples. Whenever $\mathcal{A}$ sends a tuple $(\sigma_{TEE}, outp)$ that has not been recorded by $\mathcal{S}$ to $\mathcal{F}_{blockchain}$ or an honest party, $\mathcal{S}$ simply stops the execution.

We can prove that $\mathcal{Z}$ cannot distinguish between $H_3$ and $H_2$ as follows. In $H_2$, if $\mathcal{A}$ sends forged attestations/signatures to $\mathcal{F}_{blockchain}$ or an honest party, signature verification by $\mathcal{F}_{blockchain}$ or the honest party will fail with negligible probability (as we assume the signature scheme $\Sigma$ is EU-CMA secure). If $\mathcal{Z}$ can distinguish $H_2$ from $H_3$, we can construct an adversary using $\mathcal{Z}$ and $\mathcal{A}$ to win the game of signature forgery.

Hybrid $H_4$ proceeds as $H_3$ with one modification: $\mathcal{S}$ emulates the new transaction processing. Specifically, honest parties send $newtx$ to $\mathcal{F}_{cm}$. $\mathcal{S}$ emulates messages from $\mathcal{F}_{TEE}$ and $\mathcal{F}_{blockchain}$ as in $H_3$, i.e., recording tuples $(\sigma_{TEE}, outp)$. If the party is corrupted, $\mathcal{S}$ sends $newtx(tx)$ to $\mathcal{F}_{cm}$ as $P_i$. It can be seen that the view of $\mathcal{A}$ is the same as in $H_2$, as $\mathcal{S}$ can identically emulate $\mathcal{F}_{TEE}$ and $\mathcal{F}_{blockchain}$. 

It can be seen that $H_4$ is identical to the ideal protocol. In $H_4$, while $\mathcal{S}$ interacts with $\mathcal{F}_{cm}$, it emulates $\mathcal{A}$’s view of the real-world. Now, $\mathcal{S}$ only needs to output to $\mathcal{Z}$ what $\mathcal{A}$ outputs in the real-world. Thus, there exists no environment $\mathcal{Z}$ that can distinguish between interaction between $\mathcal{A}$ and $\Pi_{cm}$, from interaction between $\mathcal{S}$ and $\mathcal{F}_{cm}$.
\end{proof}

\section{Technical overview of OptChain}\label{app:optchain}
Nguyen et al. \cite{nguyen2019optchain} proposes a transaction placement algorithm to improve the throughput of sharding blockchain, called OptChain. OptChain introduces a new sharding paradigm, in which cross-shard transaction are minimized while maintaining the load balancing among shards, resulting in almost twice faster confirmation time and throughput. Specifically, they form an optimization problem in which the algorithm determines the best shard to submit a transaction in order to minimize cross-shard transactions while guaranteeing the temporal balance, thereby shortening the confirmation time and boosting the overall system throughput. %Specifically, the ultimate goal of the transactions placement OptChain are: (1) \textbf{Fast Confirmation Time}: As major txs are same-shard txs and load are distributed evenly, txs get confirmed in much shorter time; (2) \textbf{High Throughput}: Same-shard txs require less time and communication to confirm, thus, the system throughput gets significantly boosted.
To achieve that, the algorithm optimizes the following goals:
\begin{enumerate}
    \item Minimizing cross-shard transactions: Reduce the number of cross-shard transactions by grouping related transactions into a same shard.
    \item Maintaining temporal balancing: To distribute load evenly among shards to increase parallelism and reduce queuing time. 
\end{enumerate}

On the one hand, by treating transactions as a stream of nodes in online directed acyclic graph (DAG), they propose a lightweight and on-the-fly PageRank score calculation to quantitatively measure transactions' relationship to each shard. This is referred to as the \emph{Transaction-to-Shard} (T2S) score. The T2S score between a transaction $u$ and a shard $S_i$ measures the probability that a random walk from the node $u$ ends up at some node in $S_i$, i.e., how likely the transaction should be placed into the shard without causing further cross-shard transactions.

On the other hand, they provide an estimation on the expected latency if putting a transaction into a given shard, which is called the \emph{Latency-to-Shard} (L2S) score. The L2S score estimates the processing delay when placing the transaction into a shard.

The decision making of which shard to place a transaction into is guaranteed to balance both the T2S and L2S scores. With this mechanism, OptChain has been shown to outperform Omniledger, which uses transactions' hash values to determine their output shard. Furthermore, OptChain's throughput can be scaled up linearly with the number of shards while still keeping almost the same average transaction confirmation delay.

\end{document}